%% file: main.tex
\documentclass[a4paper,11pt]{article}

\usepackage[font=bera]{chihao}

\renewcommand{\norm}[1]{\left\|#1\right\|}

\usepackage[margin=1in]{geometry}
\usepackage[affil-it]{authblk}

\newcommand{\pin}{\leftarrow}

\newcommand{\W}[1]{\+W_1\left(#1\right)}

\title{Decay of correlation for edge colorings when $q>3\Delta$}

\author[1]{Zejia Chen}
\author[1]{Yulin Wang}
\author[1]{Chihao Zhang}
\author[2]{Zihan Zhang}
\affil[1]{Shanghai Jiao Tong University}
\affil[2]{Graduate Institute for Advanced Studies, SOKENDAI}

\begin{document}
\maketitle

\begin{abstract}
    We examine various perspectives on the decay of correlation for the uniform distribution over proper $q$-edge colorings of graphs with maximum degree $\Delta$.

    First, we establish the coupling independence property when $q\ge 3\Delta$ for general graphs. Together with the recent work of Chen, Feng, Guo, Zhang and Zou (2024), this result implies a fully polynomial-time approximation scheme (\textbf{FPTAS}) for counting the number of proper $q$-edge colorings. 

    Next, we prove the strong spatial mixing property on \emph{trees}, provided that $q> (3+o(1))\Delta$. The strong spatial mixing property is derived from the spectral independence property of a version of the weighted edge coloring distribution, which is established using the matrix trickle-down method developed in Abdolazimi, Liu and Oveis Gharan (FOCS, 2021) and Wang, Zhang and Zhang (STOC, 2024).

    Finally, we show that the weak spatial mixing property holds on trees with maximum degree $\Delta$ if and only if $q\ge 2\Delta-1$.
\end{abstract}

\setcounter{tocdepth}{1}
\tableofcontents
\newpage

\input{intro}

\input{prelim}
\input{marginals}

\input{fptas}

\input{ssm}

\input{weighted}
\input{wsm}

\bibliographystyle{alpha}
\bibliography{edge_coloring}

\appendix

\input{mtd}

\end{document}

%% file: intro.tex
\section{Introduction}

Sampling from the uniform distribution of proper edge colorings received lots of attention recently, with the advent of new tools in analyzing high-dimensional distributions~\cite{DHP20,ALOG21,WZZ24,CCFV25}. A proper edge coloring of an undirected graph with maximum degree $\Delta$ is an assignment of each edge with one of $q$ colors so each adjacent edges receive different color. Clearly a proper \emph{edge} coloring can be viewed as a proper \emph{vertex} coloring on its line graph.The sampling problem is to draw a proper edge coloring uniformly. Most of previous work focuses on the mixing time of Glauber dynamics. The work of~\cite{WZZ24} established the spectral independence property, a notion to measure the correlation in high-dimensional distributions~\cite{ALO21}, whenever $q>(2+o(1))\Delta$ for general graphs and the work of~\cite{CCFV25} establish approximate tensorization of variance on trees when $q\ge \Delta+1$. Note that Glauber dynamics is known to be reducible when $q<2\Delta$ on general graphs~\cite{MJNP19} and therefore the $q>(2+o(1))\Delta$ condition for Glauber dynamics is asymptotically tight. However, it is still open whether $q\approx 2\Delta$ is the threshold for efficient sampling proper edge colorings and there are some recent attempts to design other sampling algorithm under better conditions~\cite{DKLP25}. A (almost) uniformly sampling algorithm can be turned into fully polynomial-time randomized scheme (\textbf{FPRAS}) for counting the number of proper colorings using standard reduction~\cite{JVV86}. 

In this work we examine some other aspects for the correlation of the uniform distribution on edge colorings. We first established the \emph{coupling independence} property on general graphs when $q\ge 3 \Delta$. Coupling independence~\cite{CZ23} is a notion stronger than spectral independence, and thanks to recent work of~\cite{CFGZZ24}, building on the machinery of Moitra~\cite{Moitra19}, it (together with some other properties) implies a fully polynomial-time approximate scheme (\textbf{FPTAS}) for counting proper edge colorings. This is the first \emph{determinstic} approximate counting algorithm in this range of parameters. Moreover, our proof of the coupling independence property differs from those directly derived from contractive couplings for establishing rapid mixing results.

\begin{theorem}[Informal] \label{thm:FPTAS-informal}
    If $q\geq 3\Delta$, then there exists an \textbf{FPTAS} for counting the number of proper $q$-edge colorings on any graph $G$ with maximum degree $\Delta$.
\end{theorem}
Before our work, there is no similar results tailored for counting edge colorings. The best \textbf{FPTAS} for counting edge coloring is the same as the one for general vertex coloring, which requires $q>3.618\Delta$~\cite{CV25,CFGZZ24} for sufficiently large $\Delta$.

We then study the strong spatial mixing (SSM) property for edge colorings on \emph{trees}, an important notion to measure the correlation between sites in Gibbs distributions whose definition is in \Cref{sec:prelim-decay}.


\begin{theorem}[Informal]\label{thm:SSM-informal}
    Let $T$ be a tree with maximum degree $\Delta$. If $q\geq (3+o(1))\Delta$, then the uniform distribution over $q$-colorings on $T$ exhibits strong spatial mixing with exponential decay rate.
\end{theorem}

Similar strong spatial mixing bounds on trees have been thoroughly studied for vertex colorings (e.g.~\cite{EGHSV19,CLMM23}). It is conjectured that SSM holds whenever $q\ge \Delta+1$ and in~\cite{CLMM23}, a $q\ge \Delta+3$ condition was established, almost resolving the conjecture. However, much less is known for edge colorings. 
The $q>(3+o(1))\Delta$ bound established in this work is by no means tight. We also discuss the limit of our approach and possible further improvement. On the other hand, we show that one cannot expect the strong spatial mixing property to hold when $q<2\Delta$, as we prove that $2\Delta-1$ is the threshold for \emph{weak spatial mixing}. Therefore, we conjecture that SSM holds for edge colorings on graphs with maximum degree $\Delta$ whenever $q\ge 2\Delta+\gamma$ for certain constant $\gamma$. We also show that the best bound one can expect using the analysis in this paper cannot be better than $q\approx 2.618\Delta$.

\begin{theorem}[Informal]\label{thm:WSM-informal}
    If $q\geq 2\Delta-1$, then the uniform distribution over $q$-colorings on any tree with maximum degree $\Delta$ exhibits weak spatial mixing with exponential decay rate. Otherwise, there exists a tree with maximum degree $\Delta$ such that the uniform distribution over $q$-coloring on it does not satisfy the weak spatial mixing property.
\end{theorem}
In the following, we give an overview with our technique, with an emphasis on the novelty.

\subsection{Technical contribution}

\paragraph{A new coupling strategy}

The coupling independence property is established via a new local coupling for edge colorings. Our coupling can somehow be viewed as a multi-spin version of Chen and Gu's coupling~\cite{CG24} for Holant problems with boolean domain. Their coupling, using the problem of $b$-matching as an example, begins with two instances differing at one pendant edge, or equivalently, two instances with a single constraint discrepancy. During the coupling process, the number of discrepancies can never increase but has a nonzero probability of decreasing to zero. Therefore, the coupling process terminates in expected constant number of steps. 

In the problem of edge coloring, we can design a local coupling starting from a single discrepancy so that the number of discrepancies can either increase to two, decrease to zero, or remain unchanged. We then use marginal probability bounds to control the probability of discrepancy increasing, while ensuring that the number of discrepancies decreases in expectation. 

\paragraph{Dimension reduction}

We establish the strong spatial mixing property on trees by analyzing marginal recursions, which is similar to~\cite{CLMM23}. However, unlike previous work for spin systems where the marginal on a single site is considered, we study the recursion for marginals on a ``broom'', namely all edges incident to the root. For each partial coloring on the broom, we can represent its marginal as a function of marginal probabilities of partial colorings in subtrees. However, the Jacobian matrix of this system can be as large as $q^\Delta \times \Delta q^\Delta$, and is technically very hard to analyze. Our key observation is that the Jacobian matrix is of low rank, and therefore we can apply a trace trick to bound its $2$-norm by the norm of a much smaller matrix. We call this step \emph{dimension reduction}.

\paragraph{From spectral independence to strong spatial mixing}

It is still challenging to directly bound the $2$-norm of the small matrix. We then observe that it can be written as the product of certain covariance matrices of marginal distributions on brooms. Therefore, ideally we can apply the known bounds for these covariance matrices, or equivalently the spectral independence bound for these marginals. However, these marginal probabilities are from the distribution of certain ``weighted edge colorings'' and one cannot directly apply previous spectral independence result for edge colorings. As a result, we apply the machinery of matrix trickle-down developed in~\cite{ALOG21} in the way of~\cite{WZZ24} to establish the desired spectral independence result. 

\paragraph{Top-down analysis of recursion}

There is another subtle technical issue in the above approach. When analyzing the contraction of marginal recursion, one needs to analyze the gradient / Jacobian at certain ``midpoint'' between two boundary conditions due to the application of fundamental theorem of calculus or mean-value theorem in the analysis. These midpoints, however, are not necessarily probabilities because the recursion may involve a potential function\footnote{Theses quantities are referred to as ``subdistributions'' in~\cite{CLMM23}}. In previous work, only certain marginal bounds are used to prove the contraction, and these bounds are also satisfied by the midpoints. However, in our case, we require these midpoints to satisfy refined properties, such as spectral independence, which does not hold in general. Therefore, we cannot apply the recursion in the traditional bottom-to-top manner, where one fixes the boundary value at level $L$, analyze the contraction at level $L-1$, then fixes boundary value at level $L-1$ and analyze the contraction at level $L-2$, and so on. Instead, we only fix boundary value at the leaves and analyze the composition of the recursion at each level as a whole. Therefore, we need to take ``midpoints'' only at the leaves, which defines our ``weighted edge coloring'' instance. Its spectral independence property can be established by the matrix trickle-down method, as discussed earlier.

\subsection{Organization of the paper}

After introducing the necessary preliminaries in \Cref{sec:prelim}, we give the marginal recursions and prove useful marginal bounds in \Cref{sec:marginals}. Then we present our proof of coupling independence, which implies the \textbf{FPTAS}, in \Cref{sec:FPTAS}. Strong spatial mixing on trees is proved in \Cref{sec:ssm} and the results of weak spatial mixing are proved in \Cref{sec:wsm}. In \Cref{sec:covariance}, we prove the spectral independence property for weighted edge coloring that will be used in the proof of strong spatial mixing.

%% file: prelim.tex
\newcommand{\dist}{\-{dist}}

\section{Preliminaries}\label{sec:prelim}
We use the following notations. For any $a, b\in \mathbb R$, let $
    a\wedge b\defeq \min\set{a, b}$, $
    a\vee b\defeq \max\set{a, b}$. For any two non-negative integers $a\ge b$, let $a^{\underline{b}}$ be the falling factorial, i.e. $a^{\underline{b}}=\prod_{i=a-b+1}^a i$.
Let $\!{Id}$ denote the identity matrix. For a function $f\colon \Omega\to\bb R$ defined on a finite domain $\Omega$, we use $\Big[f(x)\Big]_{x\in\Omega}$ to denote the corresponding (column) vector in $\bb R^{\Omega}$. For any set $S$ and an element $x\in S$, we write $S-x$ for $S\setminus \set{x}$.

For two probability measures $\mu, \nu$ on the same probability space $\Omega$, we define $\|\mu-\nu\|_{\-{TV}}=\frac{1}{2} \sum_{\omega \in \Omega}|\mu(\omega)-\nu(\omega)|=\sup _{A \subseteq \Omega}|\mu(A)-\nu(A)|$ for the total variation distance between $\mu, \nu$. If $\mu, \nu$ are two probability distributions on finite state spaces $\Omega_1, \Omega_2$ respectively, then we say $\omega$ is a coupling of $\mu, \nu$ when it is a joint distribution on $\Omega_1 \times \Omega_2$ with $\mu, \nu$ as its marginals, i.e. $\mu(x)=\sum_{y \in \Omega_2} \omega(x, y)$ and $\nu(y)=\sum_{x \in \Omega_1} \omega(x, y)$ for every $x \in \Omega_1, y \in \Omega_2$.

Given a graph $G=(V,E)$, for any vertex $v\in V$, let $E(v)= \set{e\in E \mid  v\in e}$ and $\deg(v)=\abs{E(v)}$ be the degree of $v$; for any edge $e \in E$, let $\deg(e)$ be the degree of $e$, which is the number of edges adjacent to $e$. Moreover, we write $\deg(G)$ for the maximum (vertex) degree in $G$.
For  two edges $e,e'\in E$, we write $\dist_G(e, e')$ for the length of the shortest path between them in $G$ (not containing $e,e'$) and $\dist_G(e, e') = \infty$ if $e$ and $e'$ is disconnected. 
Similarly, for vertices $v,v'\in V$, $\dist_G(v,v')$ is the shortest path between them in $G$ and $\dist_G(v, v') = \infty$ if $v$ and $v'$ is disconnected.

\subsection{List edge coloring}
Fix a color set $[q]=\set{1, 2, \dots, q}$ where $q\in\bb N$. Let $G=(V,E)$ be an undirected graph and $\+L=\set{\+{L}(e)\subseteq[q]\cmid e\in E}$ be a collection of color lists associated with each edge in $E$. The pair $(G, \+L)$  is an instance of list edge coloring. 

If $\+L(e)=[q]$ for any $e\in E$, we say $(G,\+L)$ is a $q$-edge coloring instance.  If $\abs{\+L(e)}\geq \deg(e)+\beta$ for any $e\in E$, we say $(G,\+L)$ is a $\beta$-extra edge coloring instance. We say $\sigma:E\rightarrow [q]$ is a proper edge coloring if $\sigma(e) \in \+L(e)$ for any $e\in E$ and $\sigma(e_1)\neq \sigma(e_2)$ for any $e_1\cap e_2 \neq \emptyset$. Let $\Omega$ denote the set of all proper edge colorings and $\mu$ be the uniform distribution on $\Omega$. 

Let $\Lambda\subseteq E$ and $\tau\in [q]^\Lambda$. We say $\tau$ is a proper partial edge coloring on $\Lambda$ if it is a proper coloring on $(G[\Lambda],\+L|_\Lambda)$ where $G[\Lambda]$ is the subgraph of $G$ induced by $\Lambda$ and $\+L|_{\Lambda}=\set{\+L(e)\in \+L\cmid E\in \Lambda}$. Let $\Omega^\tau$ be the set of all proper edge colorings on $E$ that is compatible with $\tau$, i.e. $\Omega^\tau = \set{\sigma \in \Omega\mid \tau \subset \sigma}$ . We also define $\mu^\tau$ on $\Omega$ which is supported on $\Omega^\tau$ as $\mu^\tau(\cdot) = \Pr[\sigma \sim \mu]{\sigma=\cdot \mid \tau \subset \sigma}$.
For a subset $S\subseteq E\setminus \Lambda$ and a partial coloring $\omega$ on $S$, define $\Omega^\tau_S$ as the set of all proper partial edge colorings on $S$ that is compatible with
$\tau$ and $\mu^\tau_S(\omega)=\Pr[\sigma\sim \mu]{\omega \subset \sigma \mid \tau \subset \sigma}$. Especially, when $\Lambda=\set{i}$ and $\tau(i)=c$, we write the conditional distribution and the conditional marginal distribution by $\mu^{i\leftarrow c}$ and $\mu^{i\leftarrow c}_S$. Besides, we define the color lists after pinning  $\tau$ by $\+L^\tau$ such that for any $e\in E\setminus\Lambda$, $\+L^\tau(e)=\set{c\in \+L(e)\mid \mu^\tau_e(c)>0}$ and the degree after pinning by $\deg^\tau(e) = \abs{\set{e\cup f \neq \emptyset\mid f \in E\setminus \Lambda}}$.

For a given list edge coloring instance $(G,\+L)$, let $Z_{G,\+L}(M)$ denote the number of proper colorings with the condition $M$ satisfied, (or event $M$ happens) and $\Pr[G,\+L]{M}$ denote the probability that the condition $M$ is satisfied when a proper coloring is drawn uniformly at random. For an edge set $F\subseteq E$, we usually use $c(F)$ to denote the partial coloring on $F$. With a little abuse of notation, $c(F)$ is sometimes referred to as the set of colors used on $F$.
For a color $a$, we write $a\in F, a\notin F$ as shorthands for $a\in c(F), a\notin c(F)$ respectively.

\subsection{The Wasserstein distance}
In this work, we restrict our discussions and terminologies to finite probability spaces without invoking general measure theory. 

\begin{definition}[Wasserstein distance]
    Let $\mu, \nu$ be two distributions defined on the same finite set $\Omega$ equipped with a metric $d(\cdot, \cdot)$. We define $\Gamma(\mu, \nu)$ as the set of couplings of $\mu$ and $\nu$. Then the \emph{Wasserstein ($1$-)distance} is defined by
    \[
    \W{\mu, \nu} \defeq
    \inf_{\tau\in \Gamma(\mu, \nu)} \E[(x, y)\sim \tau]{d(x, y)}.
    \]
\end{definition}

In this paper, our metric $d$ is always the Hamming distance. For two configurations $\sigma, \tau$ on $[q]^V$, their Hamming distance is defined as $d(\sigma, \tau)=\abs{\set{v \in V \mid \sigma(v)\neq \tau(v)}}$. We define the notion of coupling independence for Gibbs distribution here.
\begin{definition}[Coupling independence]
    We say a Gibbs distribution $\mu$ over $[q]^V$ satisfies $C$-coupling independence if for any two partial configurations $\sigma, \tau \in [q]^\Lambda$ on $\Lambda \subseteq V$ such that $d(\sigma, \tau)=1$,
    $$
\W{\mu^\sigma, \mu^\tau} \leq C
$$
where $\mu^\sigma$ and $\mu^\tau$ denote the Gibbs distribution conditional on $\sigma$ and $\tau$, respectively.
\end{definition}

We will use the following inequality w.r.t Wasserstein distance in later proof.
\begin{proposition}\label{prop:coupling-convex-decomposition}
    Let $\mu, \nu$ be arbitrary distributions on a common finite metric space $(\Omega, d)$.
    If there exists non-negative constants $\lambda_i, 1\le i\le k$
    and distributions $\set{\mu_i}_{1\le i\le k}$, $\set{\nu_i}_{1\le i\le k}$ on $\Omega$ such that
    \[
    \mu - \nu = \sum_{i=1}^k\lambda_i\cdot(\mu_i - \nu_i),
    \]
    where we regard both sides as functions on $\Omega$.
    Then
    \[
    \W{\mu, \nu} \le \sum_{i=1}^k \lambda_i{\W{\mu_i, \nu_i}}.
    \]
\end{proposition}
\begin{proof}
   The Kantotrovich-Rubinstein duality theorem (see Theorem 1.14 in \cite{villani2021topics}
   for the proof) states a equivalent form of Wasserstein distance:
   \begin{align*}
    \W{\mu, \nu} = \sup_{f\in L^1(\Omega)}\inner{f}{\mu - \nu},
   \end{align*}
   where $\inner{f}{\mu-\nu}\defeq\sum_{x\in \Omega} f(x)(\mu(x)-\nu(x))$ and
   $L^1(\Omega)\defeq\set{f:\Omega\to \bb R\mid \forall x, y\in \Omega : f(x)-f(y)\le d(x, y)}$
   is the space of $1$-Lipschiz functions.
   Then
   \begin{align*}
    \W{\mu, \nu} 
    =\sup_{f\in L^1(\Omega)}\inner{f}{\mu - \nu}
    =\sup_{f\in L^1(\Omega)}\inner{f}{\sum_{i=1}^k\lambda_i(\mu_i - \nu_i)}
    \le \sum_{i=1}^k\lambda_i\sup_{f\in L^1(\Omega)}\inner{f}{\mu_i - \nu_i}
    =\sum_{i=1}^k\lambda_i\W{\mu_i, \nu_i}.
   \end{align*}
\end{proof}

\subsection{Correlation Decay}\label{sec:prelim-decay}
Correlation decay refers to the phenomenon that the correlation between the color assignments of edges diminishes as their distance in the graph increases. Specifically, there are two primary notions of correlation decay: strong spatial mixing and weak spatial mixing. These two notions differ in how they measure the ``distance'' over which the correlation should decay.
\begin{definition}[Strong spatial mixing]\label{def:SSM}
    The Gibbs distribution $\mu$ of the list edge coloring instance $(G = (V,E), \+L)$ satisfies \textit{strong spatial mixing (SSM)} with exponential decay rate $1 - \delta$ and constant $C = C(q,\Delta)$ if for any $e \in E$, every subset $\Lambda \subseteq E\setminus \{e\}$ and every pair of feasible pinning $\tau_1,\tau_2$ on $\Lambda$ which differ on $\partial_{\tau_1,\tau_2} = \set{e\in \Lambda\mid \tau_1(e)\neq \tau_2(e)}$, we have that
    $$
        \|\mu_{e}^\sigma - \mu_e^\tau\|_{\-{TV}} \leq C(1 - \delta)^K
    $$
    where $K = \min_{e'\in \partial_{\tau_1,\tau_2}}\-{dist}_{G}(e,e')$.
\end{definition}
\begin{definition}[Weak spatial mixing]\label{def:WSM}
    The Gibbs distribution $\mu$ of the list edge coloring instance $(G = (V,E), \+L)$ satisfies \textit{weak spatial mixing (WSM)} with exponential decay rate $1 - \delta$ and constant $C = C(q,\Delta)$ if for any $e \in E$, every subset $\Lambda \subseteq E\setminus \{e\}$ and every pair of feasible pinning $\sigma,\tau$ on $\Lambda$, we have that
    $$
        \|\mu_{e}^\sigma - \mu_e^\tau\|_{\-{TV}} \leq C(1 - \delta)^K
    $$
    where $K = \min_{e'\in \Lambda}\-{dist}_{G}(e,e')$.
\end{definition}
\begin{remark}
    When defining Strong Spatial Mixing (SSM) and Weak Spatial Mixing (WSM) for individual problem instances, it is technically possible to choose a constant $C$ large enough to satisfy the decay inequality since the instance is finite. For SSM/WSM to meaningfully imply algorithmic tractability, these constants must hold universally for an entire class of instances, independent of the size of graph. Therefore, when we informally state that \emph{we are establishing weak / strong spatial mixing property} in this context, we refer to the existence of uniform constants $C$ and $\delta$ for a certain class of edge coloring instances even as the graph grows arbitrarily large.
\end{remark}

%% file: marginals.tex
\section{Recursion and marginals}\label{sec:marginals}
In this section, we analyze the marginal bounds on the set of edges adjacent to a given vertex in a $\beta$-extra list-edge-coloring instance under various conditions.
\subsection{Marginal bounds for general edge coloring}
We begin by examining the marginal bounds on general graphs.
\begin{lemma}\label{lem:claw-marginal-generalized}
     Given a list-edge-coloring instance $(G, \+L)$ where $G=(V, E)$,
     and a vertex $v\in V$ such that $\forall e\in E_{v}: \abs{\+L(e)}-\deg(e) \ge \beta\ge 2$.
    Then for any $v\in V$, $F\subseteq E(v)$ and color $a$:
    \[\Pr[G,\+L]{a\in c(F)} \le \frac{|F|}{\beta - 1 + |F|}\quad\mbox{and}\quad \frac{\Pr[G,\+L]{a\in c(F)}}{\Pr[G,\+L]{a \notin c(E(v))}} \le \frac{|F|}{\beta - 1}.\]
\end{lemma}
\begin{proof}
We denote the set of all proper colorings on $(G, \+L)$ by $\Omega$, and define
\[A \defeq \set{\omega\in \Omega \mid \exists e\in F: \omega(e) = a}.\]
Then define a function $\iota: (\Omega\setminus A)\times A\to \mathbb R_{\ge 0} $ such that for any $\omega'\in A$:
\begin{equation}\label{eq:marginal-sum-is-1}
\sum_{\omega\in \Omega\setminus A} \iota(\omega, \omega') = 1.
\end{equation}
To continue with the construction of $\iota$, for any $\omega' \in A$:
\[
    B(\omega')\defeq \set{\omega\in \Omega\setminus A : d(\omega, \omega') = 1}.
\]
where $d$ is the Hamming distance, i.e., the number of edges colored differently between the two colorings.
Then
\begin{align*}
\iota(\omega, \omega') \defeq
\begin{cases}
\frac{1}{|B(\omega')|},& \omega\in B(\omega')\\
0, &\text{otherwise}
\end{cases}.
\end{align*}
It is easy to verify \cref{eq:marginal-sum-is-1} holds.
Notice that if $\omega\in B(\omega')$ such that $\omega'(e) = a$ for some $e\in F$,
then $\omega$ is obtained from $\omega'$ by recoloring $e$ to another color other than $a$. 
There are at least $\beta-1$ choices according to the assumption, so $|B(\omega')|\ge \beta-1$ for any $\omega'\in A$.

Similarly, if $\omega'\in B(\omega)$, then $\omega'$ is obtained from $\omega$ by recoloring an edge $e\in F$ to $a$,
there are at most $|F|$ choices,
so we have $\sum_{\omega'} \iota(\omega, \omega') \le \frac{|F|}{\beta-1}$ for all $\omega\in \Omega\setminus A$.

Then
\begin{align*}
\Pr[G,\+L]{a\in c(F)}
&= \frac{\sum_{\omega'\in A}1}{\sum_{\omega\in \Omega\setminus A} 1 + \sum_{\omega'\in A}1}\\
&= \frac{\sum_{\omega'\in A}\sum_{\omega\in \Omega\setminus A} \iota(\omega, \omega')}
        {\sum_{\omega\in \Omega\setminus A} 1 + \sum_{\omega'\in A}\sum_{\omega\in \Omega\setminus A} \iota(\omega, \omega')}\\
&= \frac{\sum_{\omega\in \Omega\setminus A} \sum_{\omega'\in A} \iota(\omega, \omega')}
        {\sum_{\omega\in \Omega\setminus A} \left(1 + \sum_{\omega'\in A}\iota(\omega, \omega')\right)}\\
&\le \frac{|F|/(\beta-1)}{1 + |F|/(\beta-1)} = \frac{|F|}{\beta-1+|F|},
\end{align*}
proving the first part of the lemma.

Similarly, we define $\Omega_0\defeq\set{\omega\in \Omega \mid \not\exists e\in E(v) : \omega(e) = a}$.
Note that by definition, $\iota(\cdot, \cdot)$ is supported on $\Omega_0\times A$ and for any $\omega'$, $B(\omega')\subseteq \Omega_0$, so we have
\begin{align*}
\frac{\Pr[G,\+L]{a\in c(F)}}{\Pr[G,\+L]{a\notin c(E(v))}}
= \frac{\sum_{\omega'\in A}1}{\sum_{\omega\in \Omega_0} 1}
= \frac{\sum_{\omega'\in A}\sum_{\omega\in \Omega\setminus A} \iota(\omega, \omega')}
        {\sum_{\omega\in \Omega_0} 1 }.
\end{align*}
This can be rewritten as
\begin{align*}
\frac{\sum_{\omega\in \Omega_0} \sum_{\omega'\in A} \iota(\omega, \omega')}
        {\sum_{\omega\in \Omega_0} 1}
\le \frac{|F|}{\beta-1}.
\end{align*}
\qedhere
\end{proof}

Especially, when $\abs{F}=1$, we have the following corollary.
\begin{corollary}\label{cor:marginal-bound-gamma-delta}
    Let $(G=(V, E), \+L)$ be a $\beta$-extra list-edge-coloring instance. Then for any $e\in E, v\in V, a\in \+L(e)$, 
    \[
    \frac{\Pr[G,\+L]{c(e)=a}}{\Pr[G,\+L]{a\notin c(E_v)}} \le \frac{1}{\beta-1}.
    \]
\end{corollary}

\subsection{Tree recursion for edge coloring}\label{sec:recursion}
We now turn our attention to a more specialized structure: trees. Throughout this section, we fix a $\beta$-extra edge-coloring instance $(T=(V,E),\+{L})$ with root $r$. For any vertex $v \in V$, let $T_v$ be the sub-tree of $T$ rooted at $v$ and $E_{T_v}$ be the edge set of $T_v$. Let $C_v$ be the set of proper partial colorings on $E_{T_v}(v)$ and $\+D(C_v)$ be the set of all distributions on $C_v$. 

Assume $r$ has $d$ children $v_1, v_2, \dots, v_d$. For any $i\in [d]$, let $e_i = (r, v_i)$ and $T_i = T_{v_i}$. For brevity, since the color lists are fixed, we omit the color lists $\+L$ in the subscript in $\Pr[T,\+L]{\cdot}$ and $Z_{T,\+L}(\cdot)$. For any $i\in [d]$, we also write $\Pr[T_i]{\cdot}$ for $\Pr[T_i, \+L|_{T_i}]{\cdot}$ and $Z_{T_i}(\cdot)$ for $Z_{T_i,\+L|_{T_i}}(\cdot)$.

We introduce a tree recursion on the marginal distributions of partial colorings on ``brooms'', where a ``broom'' is referred to as the edge set $E_{T_v}(v)$ for a vertex $v\in V$. This recursion demonstrates how the marginal distributions on brooms propagate through the tree structure as in \Cref{fig:broom-recursion}.
\usetikzlibrary{trees}
    \tikzset{
        arn/.style = {circle, draw=black, fill=white, thick},
}
\begin{figure}[H]
\centering
\begin{tikzpicture}
    [sibling distance = 1.9cm,level distance = 1.3cm] \node [arn] {$r$}
    child[edge from parent/.style={very thick,draw}] {
        [sibling distance = 0.5cm,level distance = 0.9cm]node [arn] {$v_1$}
        child {
            [red] node [arn] {}
            edge from parent node[above left] {\small $E_{T_{v_1}}(v_1)$}
        }
        child {[white] node [red] {\tiny$\ $\Large$\dots$} }
        child {[red] node [arn] {} }
    }
    child[edge from parent/.style={very thick,draw}] {
        [sibling distance = 0.5cm,level distance = 0.9cm]node [arn] {$v_2$}
        child {[orange] node [arn] {}}
        child {[white] node [orange] {\tiny$\ $\Large$\dots$} }
        child {
            [orange] node [arn] {}
            edge from parent node[right] {\small $E_{T_{v_2}}(v_2)$}
        }
    }
    child[edge from parent/.style={very thick,draw}] {
        [white] node [black] {\Huge$\dots$}
    }
    child[edge from parent/.style={very thick,draw}] {
        [sibling distance = 0.5cm,level distance = 0.9cm]node [arn] {$v_d$}
        child {[blue] node [arn] {}}
        child {[white] node [blue] {\tiny$\ $\Large$\dots$} }
        child {
            [blue] node [arn] {}
            edge from parent node[above right] {\small $E_{T_{v_d}}(v_d)$}
        }
        edge from parent node[above right] {\small $E(r)$}
    };
\end{tikzpicture}
\caption{Brooms on a tree}
\label{fig:broom-recursion}
\end{figure}
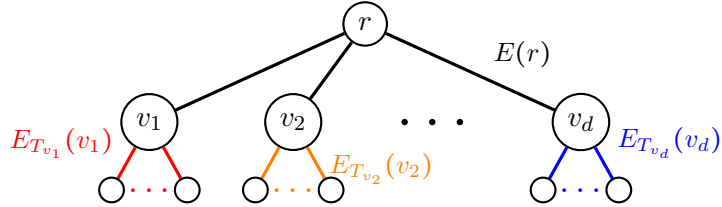

\begin{lemma}\label{lem:tree-recursion}
    Given distributions $\*{p}_i=\tp{\Pr[T_i]{c(E_{T_i}(v_i))=\tau}}_{\tau\in C_{v_i}}$ on $C_{v_i}$ for each $v_i$, we can compute the marginal distribution $\*{p}_r=\tp{\Pr[T]{c(E(r))=\pi}}_{\pi \in C_r}$ on $C_r$:
\begin{align*}
    \*p_r(\pi)=f_\pi(\set{\*{p}_i}_{i\in [d]})=\frac{\prod_i \sum_{\tau
    \in C_{v_i}: \pi(e_i)\notin \tau} \*p_{i}(\tau)}{\sum_{\rho \in C_r} \prod_i \sum_{\tau
    \in C_{v_i}: \rho(e_i)\notin \tau} \*p_{i}(\tau)}.
\end{align*}
\end{lemma}
\begin{proof}
    By the definition of marginal probabilities, we have
\begin{align*}
    \*p_r(\pi)&=\Pr[T]{c(E(r))=\pi}\\
    &=\frac{Z_T(c(E(r))=\pi)}{\sum_{\rho \in C_r} Z_T(c(E(r))=\rho)}.
\end{align*}
Since $T$ is a tree, we can further write 
\begin{align*}
    \*p_r(\pi)&=\frac{\prod_i Z_{T_i}(\pi(e_i)\not\in c(E_{T_i}(v_i)))}{\sum_{\rho \in C_r} \prod_i Z_{T_i}(\rho(e_i)\not\in c(E(v_i)))}\\
    &=\frac{\prod_i \sum_{\tau
    \in C_{v_i} : \pi(e_i)\notin \tau} Z_{T_i}(c(E_{T_i}(v_i))=\tau)}{\sum_{\rho \in C_r} \prod_i \sum_{\tau
    \in C_{v_i}: \rho(e_i)\notin \tau} Z_{T_i}(c(E_{T_i}(v_i))=\tau)}\\
    &=\frac{\prod_i \sum_{\tau
    \in C_{v_i}: \pi(e_i)\notin \tau} \Pr[T_i]{c(E_{T_i}(v_i))=\tau}}{\sum_{\rho \in C_r} \prod_i \sum_{\tau
    \in C_{v_i}: \rho(e_i)\notin \tau} \Pr[T_i]{c(E_{T_i}(v_i))=\tau}}.
\end{align*}
\end{proof}
We will regard $f=(f_\pi)_{\pi \in C_r}: \bb{R}_{\geq 0}^{C_{v_1}}\times \bb{R}_{\geq 0}^{C_{v_2}}\times \cdots \times \bb{R}_{\geq 0}^{C_{v_d}} \rightarrow \+D(C_{r})$ as a function taking inputs $\*p = (\*p_1, \*p_2, \dots, \*p_d)$ where $\*p_i\in \bb R_{\ge 0}^{C_{v_i}}$ for $i\in [d]$. Note that in some cases, $\*p$ might not encode a distribution.

\subsection{Marginal bounds propagated by recursion}\label{sec:marginal_bounds_on_trees}
Now we can give marginal bounds of probabilities that propagated by the recursion. For brevity, we define some notations of marginals for any non-zero vector $\*p_v\in \bb{R}_{\geq 0}^{C_v}$ with $v\in V$ and $c\in [q]$: Let $\*p_v(c)=\sum_{\tau \in C_v:c \in \tau} \*p_v(\tau)$, $\*p_v(\bar{c})=\sum_{\tau \in C_v:c \notin \tau} \*p_v(\tau)$, and $\*p_v(\bar{c_1},\bar{c_2})=\sum_{\tau \in C_v:c_1,c_2 \notin \tau} \*p_v(\tau)$.

Especially, for $\*p_r \in \+D(C_r)$, we define $\*p_r(i,c)=\sum_{\pi \in C_r: \pi(e_i)=c} \*p_r(\pi)$, and $\*p_r(i,c_1,j,c_2)=\sum_{\pi \in C_r: \pi(e_i)=c_1, \pi(e_j)=c_2} \*p_r(\pi)$.

By \Cref{lem:tree-recursion}, we have
\begin{equation}\label{eq:tree-recursion-short}
    \*p_r(\pi)=f_\pi(\set{\*{p}_i}_{i\in [d]})=
    \frac{\prod_i \*p_{v_i}(\ol{\pi(e_i)})}
    {\sum_{\rho \in C_r} \prod_i \*p_{v_i}(\ol{\rho(e_i)})}.
\end{equation}

We have the following lemma similar to \Cref{lem:claw-marginal-generalized}
on trees. Note that we do not require $\*p_i$'s to be distributions in \Cref{lem:marginal_bound_1}.
\begin{lemma}\label{lem:marginal_bound_1}
     Given non-zero vector $\*p_i\in \bb{R}_{\geq 0}^{C_{v_i}}$ for each $i\in [d]$ and $\*p_r = f(\{\*p_i\}_{i\in[d]})$, we have that for any color $a \in [q]$,
    \[
        \*p_r(a) \leq \frac{d}{\beta - 1 + d} \quad\mbox{and}\quad  \frac{\*p_r(a)}{\*p_r(\overline{a})} \leq \frac{d}{\beta - 1}.
    \]
    \end{lemma}
    \begin{proof}[Proof of \Cref{lem:marginal_bound_1}] The proof is similar to that of \Cref{lem:claw-marginal-generalized}.
    Define
\[A \defeq \set{\omega\in \Omega \mid \exists e\in E(r): \omega(e) = a}.\]
For any $\omega' \in A$, define
\[
    B(\omega')\defeq \set{\omega\in C_r\setminus A : d(\omega, \omega') = 1}.
\]
where $d$ is the Hamming distance, i.e., the number of edges colored differently between the two colorings.
        We define
        $$
            \iota(\omega,\omega') := \begin{cases}
                \frac{w(\omega)}{\sum_{\omega\in B(\omega')}w(\omega)}w(\omega'),&\omega\in B(\omega')
                \\ 0,& \text{otherwise}
            \end{cases}
        $$
        where $w(\omega) \defeq \prod_{i=1}^{d} \*p_i(\overline{\omega(e_i)})$. Note that $\sum_{\omega \in C_r\setminus A} \iota(\omega,\omega')=w(\omega')$.
        For any $\omega' \in A$, assuming $\omega'(e_i) = a$, by~\eqref{eq:tree-recursion-short},
        \begin{align*}
            \frac{w(\omega')}{\sum_{\omega\in B(\omega')} w(\omega)}
              &= \frac{\*p_i(\overline a)}{\sum_{a'\neq a} \*p_i(\overline{a'})}
            \\&= \frac{\sum_{\tau\in C_{v_i},a\notin \tau}\*p_i(\tau)}
              {\sum_{\tau\in C_{v_i}} \tp{\sum_{a'\neq a}\1{a\notin \tau}} \*p_i(\tau)}
            \\&\leq \frac{\sum_{\tau\in C_{v_i},a\notin \tau}\*p_i(\tau)}
              {\sum_{\tau\in C_{v_i}, a\notin\tau} \tp{\sum_{a'\neq a}\1{a\notin \tau}} \*p_i(\tau)}
            \\&\le \frac{1}{\beta-1}.
        \end{align*}
    The last inequality is because $\sum_{a'\neq a}\1{a\notin \tau}\ge \beta-1$ for all $\tau\in C_{v_i}$.
    Then we have $\sum_{\omega' \in A} \iota(\omega, \omega')\leq \frac{d}{\beta-1} w(\omega)$.    
    Therefore by~\eqref{eq:tree-recursion-short},
    \begin{align*}
\*p_r(a)
&= \frac{\sum_{\omega'\in A} w(\omega')}{\sum_{\omega\in C_r\setminus A} w(\omega) + \sum_{\omega'\in A} w(\omega')}\\
&= \frac{\sum_{\omega'\in A}\sum_{\omega\in C_r\setminus A} \iota(\omega, \omega')}
        {\sum_{\omega\in \Omega\setminus A} w(\omega) + \sum_{\omega'\in A}\sum_{\omega\in C_r\setminus A} \iota(\omega, \omega')}\\
&= \frac{\sum_{\omega\in C_r\setminus A} \sum_{\omega'\in A} \iota(\omega, \omega')}
        {\sum_{\omega\in C_r\setminus A} \left(w(\omega) + \sum_{\omega'\in A}\iota(\omega, \omega')\right)}\\
&\le \frac{d/(\beta-1)}{1 + d/(\beta-1)} = \frac{d}{\beta-1+d}.
\end{align*}
Similarly,
\begin{align*}
\frac{\*p_r(a)}{\*p_r(\bar{a})}
&= \frac{\sum_{\omega'\in A} w(\omega')}{\sum_{\omega\in C_r\setminus A} w(\omega)}
= \frac{\sum_{\omega'\in A}\sum_{\omega\in C_r\setminus A} \iota(\omega, \omega')}
        {\sum_{\omega\in \Omega\setminus A} w(\omega)}
\le \frac{d}{\beta-1}.
\end{align*}
    \end{proof}
    
\begin{lemma}\label{lem:marginal_bound_2}
    Given non-zero vector $\*p_i\in \bb{R}_{\geq 0}^{C_{v_i}}$ for each $i\in [d]$ and $\*p_r = f(\{\*p_i\}_{i\in[d]})$, we have that for any color $a\in [q]$ and $j\in [d]$,
    $$
        \*p_r(j,a) \leq \frac{\*p_j(\overline a)}{(\beta - 1)\sum_{\tau\in C_{v_j}}\*p_j(\tau)}.
    $$
\end{lemma}
\begin{proof}
        We have that
        \begin{align*}
            \*p_r(j,a) &= \frac{\displaystyle \sum_{\substack{\sigma\text{: coloring on }E(v)\setminus e_j\\a\notin \sigma}} \prod_{i\neq j}\*p_i(\overline{\sigma(v_i)})\*p_j(\overline{a})}{\displaystyle \sum_{\sigma\text{: coloring on }E(r)\setminus e_j} \prod_{i\neq j}\*p_i(\overline{\sigma(v_i)})(\sum_{a'\notin \sigma} \*p_j(\overline{a'}))}
            \\&\leq \sup_{\sigma}\frac{\displaystyle \*p_j(\overline{a})}{\displaystyle \sum_{a'\notin \sigma} \*p_j(\overline{a'})}
            \\&= \sup_{\sigma}\frac{\displaystyle \*p_j(\overline{a})}{\displaystyle(\+L(e_j) - \deg(v_j) + 1)\sum_{\tau\in C_{v_j}}\*p_j(\tau) - \sum_{a'\in \sigma} \*p_j(\overline{a'})}
            \\&\leq \frac{\*p_j(\overline{a})/(\sum_{\tau\in C_{v_j}}\*p_j(\tau))}{\+L(e_j) - \deg(v_j) - \deg(r) + 1}.
        \end{align*}
\end{proof}

%% file: fptas.tex
\newcommand{\ma}{\mu_{E-i}^{i\pin a}}
\newcommand{\mb}{\mu_{E-i}^{i\pin b}}
\newcommand{\mab}{\mu_{E-i}^{\substack{i\pin a\\j\pin b}}}
\newcommand{\mba}{\mu_{E-i}^{\substack{i\pin b\\j\pin a}}}
\newcommand{\man}{\mu_{E-i}^{\substack{i\pin a\\b\notin N}}}
\newcommand{\mbn}{\mu_{E-i}^{\substack{i\pin b\\a\notin N}}}
\renewcommand{\pab}{\mu_{E-i}^{i\pin a}(j\pin b)}
\newcommand{\pba}{\mu_{E-i}^{i\pin b}(j\pin a)}
\newcommand{\pan}{\mu_{E-i}^{i\pin a}(b\notin N)}
\newcommand{\pbn}{\mu_{E-i}^{i\pin b}(a\notin N)}
\newcommand{\sjn}{\sum_{j\in N}}
\newcommand{\sj}{\sum_{j}}

\newcommand{\zzz}{1+\sum_k(\gamma_k\vee\delta_k)}

\section{FPTAS for counting proper edge colorings on general graphs $q\ge 3\Delta$}\label{sec:FPTAS}

In this section, we prove the following main algorithmic result, which is a formal version of \Cref{thm:FPTAS-informal}.
\begin{theorem}\label{thm:FPTAS}
    Assume $\Delta\geq 4$. There exists a deterministic algorithm that outputs $\hat{Z}$ satisfying $(1-\delta)Z_{G,\+L}\leq \hat{Z}\leq (1+\delta)Z_{G,\+L}$ for any $(\Delta + 2)$-extra edge coloring instance $(G,\+L)$ with maximum degree $\Delta$ and given error bound $0<\delta <1$ in time $\tp{\frac{n}{\delta}}^{C(\Delta)}$,  where $n$ is the number of edges in $G$ and $C(\Delta)=\+O\tp{\Delta^{\Delta\log \Delta}\log \Delta}$
    is a universal constant only depends on $\Delta$.
\end{theorem}

Our key contribution is the following coupling independence result.

\begin{theorem}\label{thm:coupling-independence}
    Let $(G=(V, E), \+L)$ be a $((1+\eps)\Delta+1)$-extra list-edge-coloring instance.  Then $\mu_E$ is $\tp{1+\frac{2}{\eps}}$-coupling independent. That is, for any $i\in E$, $a, b\in \+L(i)$,
    \[
    \W{\mu_E^{i\pin a}, \mu_E^{i\pin b}} \le 1 + \frac{2}{\eps}.
    \]
\end{theorem}

We first set up our terminologies to argue about the Wasserstein distance. We define some upper bounds for the $\+W_1$ distance between $\lambda\Delta$-extra list-colorings on $\Delta$-degree graphs with $s$ edges with respect to one different pinning. We will construct recursion on these upper bounds.
An edge is \emph{pendant} if one of its endpoints has degree exactly $1$.
\begin{definition}[Universal upper bounds for coupling independence]\label{def:kappa}
Define
\[
\kappa_{s, \Delta, \lambda} \defeq 
\sup_{\substack{(G=(V, E), \+L)\\ i\in E\colon i \text{ is pendant} \\ a, b\in \+L(i), a\neq b}}
\W{
\mu_{E-i}^{i\pin a}, \mu_{E-i}^{i\pin b}
}
\]
where $(G, \+L)$ is taken over
\begin{enumerate}
    \item all graph $G=(V, E)$ such that $\deg(G)\le \Delta$, $|E|\le s$;
    \item all color lists $\+L$ such that $\forall e\in E: |\+L(e)|\ge \deg(e) + \lambda \Delta + 1$.
\end{enumerate}
\end{definition}
\begin{remark}
    It is clear from the definition that $\kappa_{s+1, \Delta, \lambda}\ge \kappa_{s, \Delta, \lambda}$
    and $\kappa_{1, \Delta, \lambda}=0$.
\end{remark}

Note that we only pin color on the edge $i$ in \Cref{def:kappa}. If we need other pinnings, we can simply consider the pinnings as deleting the pinned edges and remove the pinned color from the lists of their adjacent edges.

The main lemma of this section is a recursion for $\kappa_{s, \Delta, \lambda}$
and leads to \Cref{thm:coupling-independence} immediately.

\begin{lemma}\label{lem:kappa-recursion}
    Let $\lambda=1+\eps$ for some $\eps>0$ and $s\ge 2$. Then
    \[
    \kappa_{s, \Delta, \lambda} \le \frac{2}{2+\eps}\tp{\kappa_{s-1, \Delta, \lambda} + \frac12}.
    \]
\end{lemma}

The proof of \Cref{lem:kappa-recursion} is based on a greedy one-step coupling of two marginal distributions with different pinning on a single edge. We describe the coupling in \Cref{sec:coupling}. The proofs of \Cref{lem:kappa-recursion}, \Cref{thm:coupling-independence} and \Cref{thm:FPTAS} are in \Cref{sec:CI-proof}.

\subsection{Decomposition of Wasserstein distance}\label{sec:coupling}

The following lemma shows how we can go from $\kappa_{s, \Delta, \lambda}$ to $\kappa_{s-1, \Delta, \lambda}$
by one extra pinning.
\begin{lemma}\label{lem:s-to-s-1}
    For the instance $(G=(V, E), \+L)$, the pendant edge $i=\set{u, v}\in E$, and the colors $a, b\in \+L(i)$
    that fit into the definition of $\kappa_{s, \Delta, \lambda}$. 
    Suppose $\deg(u)=1$ and $\deg(v)\ge 2$, and $j\in N(i)$, $a, b\in \+L(j)$.
    Then
    \begin{enumerate}
        \item $\W{\mab, \mba} \le 1 + \kappa_{s-1, \Delta, \lambda}$,
        \item $\W{\mab, \mb}, \W{\mba, \ma} \le 1 + 2\kappa_{s-1, \Delta, \lambda}$,
        \item $\W{\man, \mbn} = 0$.
    \end{enumerate}
\end{lemma}
\begin{proof}
    Assume $j=\set{v,w}$. 
    \begin{enumerate}
    \item We define a new instance $\tp{G'=(V', E'=E-i), \+L'}$ by removing $i$, disconnecting $j$ from $v$ and delete $a, b$ in the color lists of edges in $N(v)$. Then $j$ becomes to a pendant edge. We have that
    \begin{align*}
        \W{\mu_{E-i; (G, \+L)}^{\substack{i\pin a\\ j\pin b}}, \mu_{E-i;(G, \+L)}^{\substack{i\pin b\\ j\pin a}}}
        &= 1 + \W{\mu_{E-i-j; (G, \+L)}^{j\pin b}, \mu_{E-i-j; (G, \+L)}^{j\pin a}}\\
        &= 1 + \W{\mu_{E'-j;(G', \+L')}^{j\pin b}, \mu_{E'-j;(G', \+L')}^{j\pin a}}.
    \end{align*}
    After the deletion of $i$ and the removal of $a$ or $b$ from the color lists of $N(v)$, the number of extra colors of each edge remains unchanged. So by \Cref{def:kappa}, the  Wasserstein distance
    $\W{\mu_{E'-j;(G', \+L')}^{j\pin b}, \mu_{E'-j;(G', \+L')}^{j\pin a}}$
    is bounded by $\kappa_{s-1, \Delta, \lambda}$, so we have
    \begin{align*}
    \W{\mu_{E-i; (G, \+L)}^{\substack{i\pin a\\ j\pin b}}, \mu_{E-i;(G, \+L)}^{\substack{i\pin b\\ j\pin a}}}
       \le 1+\kappa_{s-1, \Delta, \lambda}.
    \end{align*}

    \item By the law of total probability,
    \begin{align*}
    \mu_{E-i; (G, \+L)}^{\substack{i\pin a\\ j\pin b}} - \mu_{E-i;(G, \+L)}^{i\pin b}
    &=\mu_{E-i; (G, \+L)}^{\substack{i\pin a\\ j\pin b}}
    - \sum_{c\in\+L(j)-b}\mu_{E-i;(G, \+L)}^{i\pin b}(j\pin c)
    \mu_{E-i;(G, \+L)}^{\substack{i\pin b\\j\pin c}}\\
    &=\sum_{c\in\+L(j)-b}
    \mu_{E-i;(G, \+L)}^{i\pin b}(j\pin c)
        \tp{\mu_{E-i; (G, \+L)}^{\substack{i\pin a\\ j\pin b}}
    -
    \mu_{E-i;(G, \+L)}^{\substack{i\pin b\\j\pin c}}}.
    \end{align*}
    By \cref{prop:coupling-convex-decomposition} this means
    \begin{align}
    \W{\mu_{E-i; (G, \+L)}^{\substack{i\pin a\\ j\pin b}} , \mu_{E-i;(G, \+L)}^{i\pin b}}
    &\le
    \sum_{c\in\+L(j)-b}
    \mu_{E-i;(G, \+L)}^{i\pin b}(j\pin c)
    \W{
    \mu_{E-i; (G, \+L)}^{\substack{i\pin a\\ j\pin b}}
    ,
    \mu_{E-i;(G, \+L)}^{\substack{i\pin b\\j\pin c}}
    }
    \notag
    \\&\le
    1+
    \sum_{c\in\+L(j)-b}
    \mu_{E-i;(G, \+L)}^{i\pin b}(j\pin c)
    \W{
    \mu_{E-i-j; (G, \+L)}^{\substack{i\pin a\\ j\pin b}}
    ,
    \mu_{E-i-j;(G, \+L)}^{\substack{i\pin b\\j\pin c}}
    }.
    \label{eq:sum-abac}
    \end{align}
    For each $c\in\+L(j)-b$, we construct a new list-edge-coloring instance $(G'=(V, E'), \+L')$
    By removing $j$, appending a new edge $j'$ to $w$, and removing $b$ from the color lists of edges in $N(v)$. Then we have the following identities since the color constraints of each pair of edges
    are the same.
    \begin{align*}
    \mu_{E-i-j; (G, \+L)}^{\substack{i\pin a\\ j\pin b}}
    =
    \mu_{E'-i-j'; (G', \+L')}^{\substack{i\pin a\\ j'\pin b}}
    ,\quad
    \mu_{E-i-j; (G, \+L)}^{\substack{i\pin b\\ j\pin c}}
    =
    \mu_{E'-i-j'; (G', \+L')}^{\substack{i\pin c\\ j'\pin c}}.
    \end{align*}
    Applying these identities and triangle inequality to~\eqref{eq:sum-abac}, we have
    \begin{align*}
    \W{\mu_{E-i; (G, \+L)}^{\substack{i\pin a\\ j\pin b}} , \mu_{E-i;(G, \+L)}^{i\pin b}}
    &\le
    1+
    \sum_{c\in\+L(j)-b}
    \mu_{E-i;(G, \+L)}^{i\pin b}(j\pin c)
    \W{
    \mu_{E'-i-j'; (G', \+L')}^{\substack{i\pin a\\ j'\pin b}}
    ,
    \mu_{E'-i-j'; (G', \+L')}^{\substack{i\pin c\\ j'\pin c}}
    }
    \\&\le
    1+
    \sum_{c\in\+L(j)-b}
    \mu_{E-i;(G, \+L)}^{i\pin b}(j\pin c)
    \Bigg(
    \W{
    \mu_{E'-i-j'; (G', \+L')}^{\substack{i\pin a\\ j'\pin b}}
    ,
    \mu_{E'-i-j'; (G', \+L')}^{\substack{i\pin a\\ j'\pin c}}
    }
    \\&\;+
    \W{
    \mu_{E'-i-j'; (G', \+L')}^{\substack{i\pin a\\ j'\pin c}}
    ,
    \mu_{E'-i-j'; (G', \+L')}^{\substack{i\pin c\\ j'\pin c}}
    }
    \Bigg).
    \end{align*}
    Since we may remove the edge with the same pinning from the graph and remain the distribution unchanged,
    the two $\+W_1$ distances are both bounded by $\kappa_{s-1, \Delta, \lambda}$, proving the second part of 
    the lemma.

    \item For the third part, notice that the available colors of all edges in $E-i$ are 
    exactly the same, so ${\mu_{E-i}^{\substack{i\pin a\\b\notin N}}}={\mu_{E-i}^{\substack{i\pin b\\a\notin N}}}$
    and 
    \[
        \W{\mu_{E}^{\substack{i\pin a\\b\notin N}},\mu_{E}^{\substack{i\pin b\\a\notin N}}}
        =0.
        \qedhere
    \]
    \end{enumerate}
\end{proof}

Let $(G=(V, E), \+L)$ be a list-edge-coloring instance that fits into the constraints of $\kappa_{s, \Delta, \lambda}$ in \Cref{def:kappa}, $i\in E$ be a pendant edge, and $a, b\in \+L(i)$ be colors. We do a one-step coupling to reduce the $W_1$ distance between graphs with $s$ edges to that between graphs with $s-1$ edges using \Cref{prop:coupling-convex-decomposition}.

Denoting the two endpoints of $i$ by $u, v$, we may assume $u$ is the pendant vertex, that is $\deg(u) = 1$ without loss of generality.
    
We denote the non-empty set $E(v)-i$ by $N$, and use integers from $1$ to denote the edges in $N$ so that $N = \set{1, \ldots, d}$. For every edge $j\in N$, define
\begin{align*}
    \gamma_j\defeq \frac{\pab}{\pan}, \quad
    \delta_j\defeq \frac{\pba}{\pbn}.
\end{align*}
The following lemma describes the greedy coupling we use.
\begin{lemma}
    \begin{align*}
        \ma-\mb
          =&\sum_{j\in N}\frac{\gamma_j\wedge\delta_j}{\zzz}\Bigg(\mab - \mba\Bigg)
        \\&+\sum_{j\in N}\frac{(\gamma_j-\delta_j)\vee 0}{\zzz}\Bigg(\mab-\mb\Bigg)
           +\sum_{j\in N}\frac{(\delta_j-\gamma_j)\vee 0}{\zzz}\Bigg(\ma-\mba\Bigg)
        \\&+\frac{1}{\zzz}\Bigg(\man-\mbn\Bigg).
    \end{align*}
\end{lemma}
\begin{proof}
    The criterion of decomposing $\ma-\mb$ is whether
    $a$ or $b$ appears in $N$.
    Without loss of generality, we assume $\pan\le\pbn$,
    and denote $\pan/\pbn$ by $\alpha$.
    By the law of total probability, we have
    \begin{align}
        \ma = \sjn\pab\mab + \pan\man.
    \end{align}
    It is also clear that
    \begin{align*}
        \mb
          &= \alpha \mb + (1-\alpha) \mb
        \\&= \alpha \Bigg(\sjn\pba\mba + \pbn\mbn\Bigg)
         + (1-\alpha) \mb
        \\&= \alpha \sjn\pba\mba + \pan\mbn + (1-\alpha)\mb.
    \end{align*}
    The point of the multiplier $\alpha$ is to align the coefficients
    of $\man$ and $\mbn$, so that they don't pair with other distributions,
    and will not introduce the pinning $\cdot\notin N$, which reduces the 
    number of extra colors, into the recursion.
    By the above two decompositions,
    \begin{align}
           \ma - \mb
          =&\sjn\pab\mab - \alpha\sjn\pba\mba
            \notag
        \\&-(1-\alpha)\mb
            \notag
        \\&+\pan\Bigg(\man-\mbn\Bigg).
        \label{eq:u-v-fresh-coeff}
    \end{align}
    In general $\pab$ and $\pba$ do not equal, and we need to analyze them carefully.
    Then we can express $\pab,\pan,\pba,\pbn$ by them.
    \begin{align*}
          \pab &= \frac{\pab}{\pan+\sum_{k\in N} \ma(k\pin b)} = \frac{\gamma_j}{1+\sum_k\gamma_{k}}.
    \end{align*}
    We omitted the range of $k$ for simplicity.
    Similarly,
    \begin{align*}
             \pba = \frac{\delta_j}{1+\sum_{k}\delta_k},
       \quad \pan = \frac{1}{1+\sum_k\gamma_k},
       \quad \pbn = \frac{1}{1+\sum_k\delta_k},
       \quad \alpha = \frac{1+\sum_k\delta_k}{1+\sum_k\gamma_k}.
    \end{align*}
    Then we plug them into \cref{eq:u-v-fresh-coeff}.
    \newcommand{\zz}{1+\sum_k\gamma_k}
    \begin{align*}
           \ma - \mb
          =&\sj\frac{\gamma_j}{\zz}\mab - \sj\frac{\delta_j}{\zz}\mba
        \\&+\frac{\sum_k(\delta_k-\gamma_k)}{ 1 + \sum_k\gamma_k}\mb
           +\frac{1}{ 1 + \sum_k\gamma_k}\Bigg(\man-\mbn\Bigg)
        \\=&\sj\frac{\gamma_j\wedge\delta_j}{\zz}\Bigg(\mab - \mba\Bigg)
        \\&+\sum_j\frac{(\gamma_j-\delta_j)\vee 0}{\zz}\mab
           -\sum_j\frac{(\delta_j-\gamma_j)\vee 0}{\zz}\mba
        \\&+\frac{\sum_k(\delta_k-\gamma_k)}{\zz}\mb
           +\frac{1}{\zz}\Bigg(\man-\mbn\Bigg).
    \end{align*}
    By adding $\frac{\sum_k (\delta_k-\gamma_k)\vee 0}{\zz}\Big(\ma-\mb\Big)$ on both sides, we get
    \begin{align*}
           \frac{1+\sum_k(\gamma_k\vee\delta_k)}{\zz}\Big(\ma - \mb\Big)
          =&\sj\frac{\gamma_j\wedge\delta_j}{\zz}\Bigg(\mab - \mba\Bigg)
        \\&+\sum_j\frac{(\gamma_j-\delta_j)\vee 0}{\zz}\mab
           -\sum_j\frac{(\delta_j-\gamma_j)\vee 0}{\zz}\mba
        \\&+\sj\frac{(\delta_j-\gamma_j)\vee 0}{\zz}\ma
           {\color{purple}-\sj\frac{(\delta_j-\gamma_j)\vee 0}{\zz}\mb}
        \\&{\color{purple}+\frac{\sum_k(\delta_k-\gamma_k)}{\zz}\mb}
           +\frac{1}{\zz}\Bigg(\man-\mbn\Bigg).
    \end{align*}
    The highlighted terms sums to $-\sj\frac{(\gamma_j-\delta_j)\vee 0}{\zz}\mb$.
    So we can pair the terms to get
    \begin{align*}
           \frac{1+\sum_k(\gamma_k\vee\delta_k)}{\zz}\Big(\ma - \mb\Big)
          =&\sj\frac{\gamma_j\wedge\delta_j}{\zz}\Bigg(\mab - \mba\Bigg)
        \\&+\sum_j\frac{(\gamma_j-\delta_j)\vee 0}{\zz}\Bigg(\mab-\mb\Bigg)
           +\sum_j\frac{(\delta_j-\gamma_j)\vee 0}{\zz}\Bigg(\ma-\mba\Bigg)
        \\&+\frac{1}{\zz}\Bigg(\man-\mbn\Bigg).
    \end{align*}
    Finally, by dividing $\frac{\zzz}{\zz}$, 
    we get the decomposition
    \begin{align*}
        \ma-\mb
          =&\sj\frac{\gamma_j\wedge\delta_j}{\zzz}\Bigg(\mab - \mba\Bigg)
        \\&+\sum_j\frac{(\gamma_j-\delta_j)\vee 0}{\zzz}\Bigg(\mab-\mb\Bigg)
           +\sum_j\frac{(\delta_j-\gamma_j)\vee 0}{\zzz}\Bigg(\ma-\mba\Bigg)
        \\&+\frac{1}{\zzz}\Bigg(\man-\mbn\Bigg).
    \end{align*}
\end{proof}

\subsection{Proof of main theorems}\label{sec:CI-proof}

Now we prove \Cref{lem:kappa-recursion}, which provides a recursion for $\kappa_{s,\Delta,\lambda}$.

\begin{proof}[Proof of \Cref{lem:kappa-recursion}]
    We consider every $\lambda\Delta+1$-extra color instance $\tp{G=(V,E),\+L}$ with a pendant edge $i=\set{u,v}$ such that $\deg(G)\le \Delta$, $\abs{E}\le s$. Suppose $\deg(u)=1$. If $\deg(v)=1$, then $\W{\ma,\mb}=0$. In the following we assume $\deg(v)\ge 2$. 

An application of \Cref{prop:coupling-convex-decomposition} gives
\begin{align*}
    \W{\ma, \mb}
      \leq&\sj\frac{\gamma_j\wedge\delta_j}{\zzz}\W{\mab, \mba}
    \\&+\sum_j\frac{(\gamma_j-\delta_j)\vee 0}{\zzz}\W{\mab, \mb}
       +\sum_j\frac{(\delta_j-\gamma_j)\vee 0}{\zzz}\W{\ma, \mba}
    \\&+\frac{1}{\zzz}\W{\man, \mbn}.
\end{align*}

By \Cref{lem:s-to-s-1},
\begin{align*}
    \W{\ma, \mb}&\le 1-\frac{1}{\zzz} + \frac{\sum_j(\gamma_j\wedge\delta_j) + 2|\gamma_j-\delta_j|}{\zzz}\kappa_{s-1,\Delta,\lambda}
    \\&\le 1-\frac{1}{\zzz} + \frac{\sum_j2(\gamma_j\vee\delta_j)}{\zzz}\kappa_{s-1,\Delta, \lambda}.
\end{align*}
Finally, the bound $\forall j\in N : \gamma_j,\delta_j\le \frac{1}{\lambda\Delta}$ from \Cref{cor:marginal-bound-gamma-delta} gives
\begin{align*}
    \W{\ma, \mb} \le \frac{1}{1+\lambda} + \frac{2/\lambda}{1+1/\lambda}\kappa_{s-1,\Delta, \lambda}
    \le \frac{2}{2+\eps}\Big(\frac12 + \kappa_{s-1,\Delta,\lambda}\Big).
\end{align*}
Taking the supremum as in \Cref{def:kappa} proves the lemma.
\end{proof}

\begin{proof}[Proof of \Cref{thm:coupling-independence}]
    \Cref{lem:kappa-recursion} shows that $\sup_{s}\kappa_{s, \Delta, \lambda}\le \frac{1}{2+\eps}(1+2/\eps) = 1/\eps$.
    And
    \begin{align*}
        \W{\mu_{E}^{i\pin a}, \mu_{E}^{i\pin b}} \le 1 + \W{\mu_{E-i}^{i\pin a}, \mu_{E-i}^{i\pin b}}.
    \end{align*}
    To reduce to the pendant edge case,
    we may break $i$ into two pendant edges $i_1, i_2$,
    connected to the two endpoints of $i$ respectively,
    the color lists of $i_1, i_2$ are the same as $i$.
    We denote the new coloring instance by $(G',\+L')$.
    Then we have
    \begin{align*}
    \mu_{E-i      ; (G , \+L )}^{i\pin a} =
    \mu_{E-i_1-i_2; (G', \+L')}^{\substack{i_2\pin a \\ i_1\pin a}},
    \quad
    \mu_{E-i      ; (G , \+L )}^{i\pin b} =
    \mu_{E-i_1-i_2; (G', \+L')}^{\substack{i_2\pin b \\ i_1\pin b}}.
    \end{align*}
    So by triangle inequality, we have
    \begin{align*}
    &\phantom{{}={}}\W{\mu_{E;(G,\+L)}^{i\pin a}, \mu_{E;(G,\+L)}^{i\pin b}}
    \\&\le
    1 + \W{\mu_{E-i; (G,\+L)}^{i\pin a}, \mu_{E-i;(G,\+L)}^{i\pin b}}
    \\&=
    1 + 
    \W{
    \mu_{E-i_1-i_2; (G', \+L')}^{\substack{i_2\pin a \\ i_1\pin a}},
    \mu_{E-i_1-i_2; (G', \+L')}^{\substack{i_2\pin b \\ i_1\pin b}}
    }
    \\&\le
    1 + 
    \W{
    \mu_{E-i_1-i_2; (G', \+L')}^{\substack{i_2\pin a \\ i_1\pin a}},
    \mu_{E-i_1-i_2; (G', \+L')}^{\substack{i_2\pin b \\ i_1\pin a}}
    }
    +
    \W{
    \mu_{E-i_1-i_2; (G', \+L')}^{\substack{i_2\pin b \\ i_1\pin a}},
    \mu_{E-i_1-i_2; (G', \+L')}^{\substack{i_2\pin b \\ i_1\pin b}}
    }
    \\&\le
    1 + \frac{2}{\epsilon}.
    \end{align*}
\end{proof}

\Cref{thm:coupling-independence} shows that any $((1+\eps)\Delta+1)$-extra list-edge-coloring instance is $\tp{1+ \frac{2}{\eps}}$-coupling independent. This is because adding additional pinnings can be viewed as generating a new $((1+\eps)\Delta+1)$-extra list-edge-coloring instance by deleting the pinned edges and removing the corresponding colors from the lists of their adjacent edges. 

The work of~\cite{CFGZZ24} designs an \textbf{FPTAS} for counting the partition function of any Gibbs distribution of permissive spin systems that is marginally bounded and coupling independent. A spin system is specified by a $4$-tuple $S=(G = (V,E),q, A_E, A_V)$ where the state space is $[q]^V$ and the weight of a configuration is characterized by the matrices $A_E \in \bb{R}_{\geq 0}^{q\times q}$ and $A_V \in \bb{R}_{\geq 0}^q$. The Gibbs distribution is defined by:
\[
\mu(\sigma)\propto w(\sigma):=\prod_{u,v\in E} A_E(\sigma(u),\sigma(v))\prod_{v\in V} A_V(\sigma(v)).
\]
The normalizing factor of $\mu$ is called the partition function $Z:=\sum_{\sigma\in [q]^V}w(\sigma)$.

We say $S$ is permissive if for any partial configuration $\tau \in [q]^\Lambda$ with $\Lambda \subseteq V$, the conditional partition function $Z^\tau = \sum_{\sigma:\tau\subset \sigma} w(\sigma) >0$. For $\tau\in [q]^\Lambda$ with $\Lambda \subset V$, let $\mu_v^\tau$ be the marginal distribution on $v\in V\setminus \Lambda$ conditional on the partial configuration $\tau$. We say $\mu$ is $b$-marginally bounded if for any partial configuration $\tau \in [q]^\Lambda$ with $\Lambda \subseteq V$, any vertex $v \in V\setminus \Lambda$ and $c\in [q]$ such that $\mu_v^\tau>0$, we have $\mu_v^\tau\geq b$.

The main result of~\cite{CFGZZ24} that we will use is as follows.
\begin{theorem}[\cite{CFGZZ24}]\label{thm:ci-FPTAS}
    Let $q \geq 2, b>0, C>0, \Delta \geq 3$ be constants. There exists a deterministic algorithm such that given a permissive spin system $\mathcal{S}=\left(G, q, A_E, A_V\right)$ and error bound $0<\varepsilon<1$, if the Gibbs distribution of $\mathcal{S}$ is b-marginally bounded and satisfies $C$-coupling independence, and the maximum degree of $G$ is at most $\Delta$, then it returns $\hat{Z}$ satisfying $(1-\varepsilon) Z \leq \hat{Z} \leq(1+\varepsilon) Z$ in time $\left(\frac{n}{\varepsilon}\right)^{f(q, b, C, \Delta)}$, where $f(q, b, C, \Delta)=\Delta^{\+O\left(C\left(\log b^{-1}+\log C+\log \log \Delta\right)\right)} \log q$ is a constant.
\end{theorem}
In our setting, the spin system is list-edge-coloring instance. Be careful with the parameters since the spins are on edges of the graph. In order to prove \Cref{thm:FPTAS}, we need the marginal lower bound from \cite{GKM15}.
\begin{lemma}[Corollary of Lemma 3 in \cite{GKM15}] \label{lem:marginal-bound-gkm}
Fix a $\beta$-extra edge coloring instances $(G,\+L)$ of maximum degree $\Delta$ with Gibbs distribution $\mu$. For any partial coloring $\tau$ on $\Lambda \subseteq E$, any $e \in E\setminus \Lambda$, and $c \in \+L^\tau(e)$, it holds that
	\[
		\mu^{\tau}_e(c) \geq \frac{\tp{1-\frac{1}{\abs{\+L^\tau(e)}-\deg^\tau(e)}}^{\deg^\tau(e)}}{\abs{\+L^\tau(e)}}\ge \frac{(1-\frac{1}{\beta})^{2\Delta-2}}{\beta+2\Delta-2} .
	\]
 \end{lemma}
Now we give the proof of our main theorem in this section.
\begin{proof}[Proof of \Cref{thm:FPTAS}]
It is easy to verify that any $\beta$-extra edge coloring instance with $\beta \geq 1$ is permissive. From \Cref{thm:coupling-independence} and \Cref{lem:marginal-bound-gkm}, we know that the Gibbs distribution $\mu$ of a $(\Delta + 2)$-extra edge coloring instance $(G,\+L)$ is $\tp{1+ 2\Delta}$-coupling independent and $b$-marginally bounded where $b=\frac{(1-\frac{1}{\Delta + 2})^{2\Delta-2}}{3\Delta}$ is $\Omega\tp{\frac{1}{\Delta}}$. Then by \Cref{thm:ci-FPTAS}, there exists a deterministic algorithm that outputs $\hat{Z}$ satisfying $(1-\delta)Z_{G,\+L}\leq \hat{Z}\leq (1+\delta)Z_{G,\+L}$ in time $\tp{\frac{n}{\delta}}^{C(\Delta)}$ where $n$ is the number of edges in $G$ and $C(\Delta)=\+O\tp{\Delta^{\Delta\log \Delta}\log \Delta}$.
\end{proof}

%% file: ssm.tex
\section{Strong spatial mixing for edge colorings on trees when $q>(3+o(1))\Delta$}\label{sec:ssm}

In this section, we prove our theorem for strong spatial mixing.

\begin{theorem}\label{thm:SSM}
    Given a $\beta$-extra edge coloring instance $(G,\+L)$ where $G$ is a tree of maximum degree $\Delta$, the uniform distribution on such instance exhibits strong spatial mixing with exponential decay rate $1-\delta$ and constant $C = \max\{32q^{\frac{\Delta + 2}2}\Delta^2(1-\delta)^{-3},(1-\delta)^{-4}\}$ if $\beta > \max\{\Delta + 50, (1 + \eta_\Delta)\Delta + 1\}$, where
    \[
         \delta = \frac{(1 + (\beta - 1 -\Delta)/\Delta)^2 - (1+\eta_{\Delta})^2}{2(1 + (\beta - 1 -\Delta)/\Delta)^2}, \eta_{\Delta} = \+\+O\tp{\frac{\log^2 \Delta}{\Delta}}.
    \]
    Specifically, if $\beta = \tp{1 + \frac{\log^3 \Delta}{\Delta}}\Delta + 50$, then $\delta \approx \frac{\log^3 \Delta}{\Delta}$.
 \end{theorem}
 

We already introduced the recursion for marginal probabilities of edge colorings on trees and derived certain marginal bounds in \Cref{sec:marginals}.  We will then analyze its Jacobian matrix in \Cref{sec:jacobian}. Using the bounds on the norm of the Jacobian matrix, we prove \Cref{thm:SSM} in \Cref{sec:contraction}. Finally, we discuss the limit of our approach and possible further improvement in \Cref{sec:limit}. A key ingredient in our bounds for the norm of Jacobian matrix is a bound for certain covariance matrices, which is addressed in \Cref{sec:covariance}. 

\subsection{Upper bound the 2-norm of Jacobian}\label{sec:jacobian}

\subsubsection{The Jacobian}
Recall the recursion $f$ for marginals introduced in \Cref{sec:recursion}. We regard $f=(f_\pi)_{\pi \in C_r}: \bb{R}_{\geq 0}^{C_{v_1}}\times \bb{R}_{\geq 0}^{C_{v_2}}\times \cdots \times \bb{R}_{\geq 0}^{C_{v_d}} \rightarrow \+D(C_{r})$
as a function taking inputs $\*p = (\*p_1, \*p_2, \dots, \*p_d)$ where $\*p_i\in \bb R^{C_{v_i}}$ for $i\in [d]$. The Jacobian of $f$ is a matrix $(\+J f)(\*p) \in \bb{R}^{C_r\times \bigcup_{i\in [d]} C_{v_i}}$. Since $C_{v_i}$'s are disjoint, for every $\tau \in \bigcup_{i\in [d]} C_{v_i}$, we will denote it by $(i,\tau)$ if $\tau\in C_{v_i}$ for clarity. Therefore, 
\begin{equation*}
    (\+J f)_{\pi,(i,\tau)}(\*p)= \frac{\partial f_\pi}{\partial \*p_{i}(\tau)}.
\end{equation*}
For each $i\in [d]$, define the matrix $\+J_i \in \bb{R}^{C_r\times C_{v_i}}$ with entries
\begin{equation*}
    (\+J_i f)_{\pi,\tau}(\*p) = (\+J f)_{\pi,(i,\tau)}(\*p)
\end{equation*}
for every $\pi\in C_r$ and $\tau\in C_{v_i}$. 

\bigskip
We can write $\+J_if$ in a compact way. 
\begin{proposition} 
Let $\*p_r = f(\*p)$. Then
\begin{align*}
    (\+J_i f)(\*p)=\sum_{c\in \+L(e_i)}\*a_{i,c} \*b_{i,c}^{\top},
\end{align*}
where $\*a_{i,c^*}=f(\*p) \odot \Big[\1{\pi(e_i)=c^*}-\sum_{\pi' \in C_{r}: \pi'(e_i)=c^*}\*p_r(\pi')\Big]_{\pi\in C_r}$
and $\*b_{i,c^*}=\Big[\frac{\1{c^*\notin \tau}}{\sum_{\tau' \in C_{v_i}: c^*\notin \tau'}\*p_i(\tau')} \Big]_{\tau\in C_{v_i}}$. \footnote{We write $\*u \odot \*v$ for their Hadamard product (entry-wise product).}
\end{proposition}

\begin{proof}
    Let $q_{i,\rho}=\sum_{\tau
    \in C_{v_i}: \rho(e_i)\notin \tau} \*p_{i}(\tau)$.
    For any $\pi \in C_r$, we write $f_\pi$ as a function of $q_{i,\rho}$ 
    \begin{equation*}
        f_\pi = \frac{\prod_i q_{i,\pi}}{\sum_{\rho \in C_r} \prod_i q_{i,\rho}}.
    \end{equation*}
    Then we can compute
    \begin{equation*}
        \frac{\partial f_\pi}{\partial q_{i,\rho}} = \frac{1}{q_{i, \rho}}\tp{\1{\rho=\pi}-f_\rho}f_\pi.
    \end{equation*}
    Therefore,
    \begin{equation*}
          (\+J f)_{\pi,(i,\tau)}(\*p)=\frac{\partial f_\pi}{\partial \*p_i(\tau)}=\sum_{\rho\in C_r} \frac{\partial f_\pi}{\partial q_{i,\rho}} \frac{\partial q_{i,\rho}}{\partial \*p_i(\tau)}
         =\sum_{\rho\in C_r}\frac{1}{q_{i, \rho}}(\1{\rho=\pi}-f_\rho)f_\pi\cdot \1{\rho(e_i)\notin \tau} .
    \end{equation*}
We write $\+J_i f$ explicitly:
\begin{align*}
(\+J_i f)(\*p)
=\operatorname{diag}(f(\*p))\cdot\sum_{\rho \in C_r}\frac{1}{q_{i,\rho}}\cdot\Big[\1{\pi=\rho}-f_\rho(\*p)\Big]_{\pi \in C_r} \Big[\1{\rho(e_i)\notin \tau}\Big]_{\tau\in C_{v_i}}^\top .
\end{align*}
Noting that $q_{i,\rho}=\Pr[T_i]{\rho(e_i)\notin c(E_{T_i}(v_i))}$ only relies on the color $\rho(e_i)$, we have
\begin{align*}
    (\+J_i f)(\*p) = \operatorname{diag}(f(\*p))\sum_{c^*\in L({e_i})} \Big[\1{\pi(e_i)=c^*}-\sum_{\pi' \in C_{r} :\pi'(e_i)=c^*}\*p_r(\pi)\Big]_{\pi \in C_r}\Big[\frac{\1{c^*\notin \tau}}{\sum_{\tau' \in C_{v_i} :c^*\notin \tau'}\*p_i(\tau')}\Big]_{\tau\in C_{v_i}}^\top.
\end{align*}
\end{proof}

A well-known trick in the analysis of decay of correlation is to apply a potential function on the marginal recursion to amortize the contraction rate. Given an increasing  potential function $\phi:[0,1]\rightarrow \bb{R}$, we define $f^\phi$ such that for any $\pi \in C_r$ and $\*m \in \bb{R}^{C_{v_1}\times C_{v_2}\times \dots \times C_{v_d}}$,
\begin{equation*}
f^{\phi}_\pi(\*m)=\phi\left(f_\pi\left(\left(\phi^{-1}(\*m_{1}),\phi^{-1}(\*m_{2}),\dots, \phi^{-1}(\*m_{d})\right)\right)\right).
\end{equation*}

As a result, the Jacobian of $f^\phi$ can be obtained by the chain rule and the inverse function theorem as follows.
\begin{proposition}\label{prop:eq_of_Ji}
Given a smooth increasing function $\phi:[0,1]\rightarrow \bb{R}$ with derivative $\Phi = \phi'$, let $\*p = \phi^{-1}(\*m)$. Then we have
\begin{align*}
    (\+J_i f^{\phi})(\*m)=\sum_c (\Phi(f(\*p))\odot \*a_{i,c})(\*b_{i,c} \odot \Phi^{-1}(\*p_i))^{\top}.
\end{align*}
\end{proposition}
Taking $\Phi(x)=\frac{1}{\sqrt{x}}$, we have
\begin{align*}
    (\+J_i f^{\phi})(\*m)=\sum_c \*a^{\phi}_{i,c} (\*b^{\phi}_{i,c})^{\top},
\end{align*}
where $\*a^{\phi}_{i,c^*}=\sqrt{f(\*p)}\odot \Big[\1{\pi(e_i)=c^*}-\sum_{\pi' \in C_{r} : \pi'(e_i)=c^*}\*p_r(\pi')\Big]_{\pi\in C_r}$ and $\*b^{\phi}_{i,c^*}=\Big[\frac{\1{c^*\notin \tau}\sqrt{\*p_i(\tau)}}{\sum_{\tau' \in C_{v_i}: c^*\notin \tau'}\*p_i(\tau')}\Big]_{\tau\in C_{v_i}}$.

\subsubsection{Bounding $\norm{(\+J f^{\phi})(\*p)}_2$}

In this section, we aim to derive an upper bound for the 2-norm of the Jacobian of the tree recursion.
For brevity, we follow some notations which is defined in \Cref{sec:marginal_bounds_on_trees} of marginal probabilities w.r.t $\*p_i$ and $\*p_r$.

Also we introduce some notations of matrices used in later proof.
\begin{definition}\label{def:covariance}
    For a distribution $\*p$ over proper colorings on a broom $E_{T_v}(v)=\set{e_1,\dots,e_m}$, we define $X_v=\set{(i,c):i\in [m], c\in \+L(e_i)}$ and its (local) covariance matrix $ \!{Cov}(\*p)\in \bb{R}^{X_v\times X_v}$ with entries:
    $$ \!{Cov}(\*p)((i,c_1),(j,c_2)) = \sum_{\tau \in C_v:\tau(e_i)=c_1 \& \tau(e_j)=c_2} \*p(\tau)-\tp{\sum_{\tau \in C_v:\tau(e_i)=c_1}\*p(\tau)}\tp{\sum_{\tau \in C_v:\tau(e_j)=c_2}\*p(\tau)}.$$
\end{definition}

\begin{definition}\label{def:diag_mean}
    For a distribution over colorings $\*p$ on a broom $E_{T_v}(v)=\set{e_1,\dots,e_m}$, we define $X_v=\set{(i,c):i\in [m], c\in \+L(e_i)}$ and  its diagonal matrix of mean vector $\Pi(\*p)\in \bb{R}^{X_v\times X_v}$ as follows,
    $$
        \Pi(\*p) = \-{diag}\set{\sum_{\tau\in C_v: \tau(e_i) = c}\*p(\tau)}_{(i,c)\in X_v}.
    $$
\end{definition}
Now we define the notion of spectral independence on a broom.
\begin{definition}\label{def:sepctral_independence}
    For any distribution over colorings $\*p$ on a broom $E(v) = \{e_1,...,e_m\}$, we say $\*p$ is $C$-spectrally independent if it holds that
    $$
         \!{Cov}(\*p) \preceq C\cdot \Pi(\*p).
    $$
\end{definition}
The following is the main result in this section.
\begin{condition}[marginal bound]\label{cond:marginal}
    For any $i\in [d]$, $\*p_i$ is a distribution on $C_{v_i}$ such that for any color $a$
\[
\*p_i(a)\leq \frac{\abs{E_{T_i}(v_i)}}{\beta -1 +\abs{E_{T_i}(v_i)}},
\]
and for $\*p_r = (f_\pi(\*p_1,\*p_2, \dots, \*p_d))_{\pi \in C_v}$,
\[
\frac{\*p_r(i,a)}{\*p_i(\bar{a})}\leq \frac{1}{\beta -1},
\]
where $\beta \geq (1+o(1))\Delta$ .
\end{condition}

\begin{theorem}\label{thm:bound_Jacobian}
    For any $i\in [d]$, $\*p_i = \phi^{-1}(\*m_i)$ and $\*p_r = f((\*p_i)_{i\in[d]})$ satisfy \Cref{cond:marginal} and $(1 + \eta)$-spectrally independent. Then $\beta \geq 1 + \frac{(1 + \eta)\Delta}{\sqrt{1 - 2\delta}}$ implies that 
    $$
        \norm{\+J f^\phi(\*m)}_2 \leq \frac{1-\delta}{\sqrt \Delta}
    $$
    where $\*m = \phi(\*p)$ and $\phi(x) = 2\sqrt{x}$.
\end{theorem}

We will prove the theorem in \Cref{sec:bound_Jac} after introducing our key reduction in \Cref{sec:dim}.

\subsubsection{Dimension reduction}\label{sec:dim}
By the definition of 2-norm, we have that $\norm{(\+J f^\phi)(\*p)}_2 = \sqrt{\lambda_{\max}( (\+J f^\phi)(\*p)(\+J f^\phi)(\*p)^\top )}$. Let $\*A := (\+J f^\phi)(\*p)(\+J f^\phi)(\*p)^\top$. We have that
\begin{align*}\label{eq:2norm_J}                 
    \*A
    = \sum_{i=1}^d  (\+J_i f^\phi)(\*p)(\+J_i f^\phi)(\*p)^\top
    = \sum_{i=1}^d \sum_{c_1,c_2\in \+L({e_i})} \inner{\*b_{i,c_1}^\phi}{\*b_{i,c_2}^\phi} \*a_{i,c_1}^\phi (\*a_{i,c_2}^\phi)^\top.
\end{align*}
The last equation simply follows from~\Cref{prop:eq_of_Ji}. The above calculation suggests that although the dimension of $\*A$ is exponential in $d$, its rank is polynomial in $d$. In the following, we will find a much smaller matrix which can be used to upper bound $\*A$. The idea is to use the trace method, namely to study $\Tr\tp{\*A^k}$. We have the following lemma.


\begin{lemma}\label{lem:clac_lammax}
    For any positive semi-definite matrix $\*M \in \mathbb R^{n\times n}$, we have that
    \[
        \lambda_{\max} (\*M) = \lim_{k\to \infty}
        \tp{\Tr(\*M^k)}^{\frac 1k}.
    \]
\end{lemma}
\begin{proof}[Proof of~\Cref{lem:clac_lammax}]
    Assume that $\lambda_1,...,\lambda_n$ are eigenvalues of $\*M$ and $0\leq \lambda_1 \leq \lambda_2 \leq ... \leq \lambda_n$.
    $\*M$ can be factored as $\*Q\*\Lambda\*Q^{-1}$ where $\*\Lambda$ is a diagonal matrix satisfying $\*\Lambda(i,i) = \lambda_i$.
    Therefore,
    \begin{equation*}
        \lim_{k\to \infty}\Tr\tp{\*M^k}^{\frac 1k}
        = \lim_{k\to \infty}\Tr\tp{\*Q\*\Lambda^k\*Q^{-1}}^{\frac 1k}
        = \lim_{k\to \infty}\Tr\tp{\*\Lambda^k\*Q^{-1}\*Q}^{\frac 1k}
        = \lim_{k\to \infty} \tp{\sum_{i=1}^n \lambda_i^k}^{\frac 1k}
        =\lambda_n.
        \qedhere
    \end{equation*}
\end{proof}
To simplify notations, we let
\[
V(i,z_1,c_1,z_2,c_2) \defeq \inner{\*b_{i,c_1}^\phi}{\*b_{i,c_2}^\phi}
    \tp{\1{z_1=c_1} - \*p_r(i,c_1)}\tp{
    \1{z_2=c_2} - \*p_r(i,c_2)}.
\]
Then we can write $\*A$ explicitly.
\begin{equation}\label{eq:equation_A}
    \*A(\pi,\tau) = \sqrt{f(\*p)(\pi)f(\*p)(\tau)} \sum_{i=1}^d \sum_{c_1,c_2\in \+L(e_i)} V(i,\pi(e_i),c_1,\tau(e_i),c_2).
\end{equation}
Let $g_k^\pi(i,c)$ denote
\begin{align*}
    &\sum_{\substack{\tau_1,...,\tau_{k-1}\in C_r\\ \tau_0=\pi}} \prod_{j=1}^{k-1}f(\*p)(\tau_j)
    \sum_{i_1,...,i_{k-1}\in[d],i_k=i}
    \sum_{\substack{c_{1,1}\in \+L(e_{i_1})\\...\\c_{1,k}\in \+L(e_{i_k})}}
    \sum_{\substack{c_{2,1}\in \+L(e_{i_1})\\...\\c_{2,k-1}\in \+L(e_{i_{k-1})}\\c_{2,k}=c}}
    \prod_{j=1}^{k-1} V(i_j,\tau_{j-1}(i_j),c_{1,j},\tau_j(i_j),c_{2,j})
    \\&\times \inner{\*b_{i,c_{1,k}}^\phi}{\*b_{i,c}^\phi}(\1{\tau_{k-1}(i) = c_{1,k}} -  \*p_r(i,c_{1,k})).
\end{align*}
We omit $\pi$ in $g_k^\pi$ for brevity.
Then we have that for any $\pi \in C_r$
\begin{align}\label{eq:trace_of_A}
    \*A^k(\pi,\pi) = f(\*p)(\pi)\sum_{i=1}^d
    \sum_{c\in \+L(e_i)} g_k(i,c) (\1{\pi(i) = c} -  \*p_r(i,c)).
\end{align}
Fix $\pi$, then we will show that $\{g_k\}_{k\geq 1}$ can be computed recursively, which gives a simple representation of $\*A^k(\pi,\pi)$. Let $X=\set{(i,c)|i\in [d], c\in \+L(e_i)}$ be the set of all feasible edge-color pairs. 
\begin{lemma}\label{lem:recursion_for_g}
    If $\*B(\*p)\in \mathbb R^{X\times X}$ satisfies that
    \begin{align*}
        \*B(\*p)((i,c_2),(j,c_4)) &= 
        \sum_{c_3 \in \+L(e_j)}
        \frac{\*p_j(\bar{c_3},\bar{c_4})}{\*p_j(\bar{c_3})\*p_j(\bar{c_4})}\times \tp{\*p_r(j,c_3,i,c_2) - \*p_r(j,c_3)\*p_r(i,c_2)}.
    \end{align*}
    Then we have that $g_k^\top = \alpha^\top_\pi \*B^{k-1}$ where
    \begin{align*}
        \alpha_\pi(i,c_2) &= \sum_{c_1 \in \+L(e_i)}
         \frac{\*p_i(\bar{c_1},\bar{c_2})}{\*p_i(\bar{c_1})\*p_i(\bar{c_2})}
        \times (\1{\pi(i) = c_1} - \*p_r(i,c_1)).
    \end{align*}
\end{lemma}
\begin{proof}[Proof of \Cref{lem:recursion_for_g}]
    For any $k > 1$, we expand one layer of summation and get
    \begin{align*}
        g_k(i,c) &= \sum_{\tau_{k-1}\in C_r} f(\*p)(\tau_{k-1})\sum_{i_{k-1}\in [d]} \sum_{c_{1,k}\in \+L(e_i)} \sum_{c_{2,k-1}\in \+L(e_{i_{k-1}})}
        g_{k-1}(i_{k-1},c_{2,k-1}) 
        \\&\quad\times \inner{\*b_{i,c_{1,k}}^\phi}{\*b_{i,c}^\phi}(\1{\tau_{k-1}(i) = c_{1,k}} - \*p_r(i,c_{1,k}))
        (\1{\tau_{k-1}(i_{k-1}) = c_{2,k-1}} - \*p_r(i_{k-1}, c_{2,k-1}))
        \\&= \sum_{c_1\in \+L(e_i)}\inner{\*b_{i,c_1}^\phi}{\*b_{i,c}^\phi}\sum_{j\in[d]}\sum_{c_{2,k-1}\in \+L(e_{j})}(\*p_r(i,c_1,j,c_{2,k-1})-\*p_r(i,c_1)\*p_r(j,c_{2,k-1})) g_{k-1}(j,c_{2,k-1}).
    \end{align*}
    Recall that
    $$ 
        \inner{\*b_{i,c_1}^\phi}{\*b_{i,c}^\phi} = 
        \frac{\*p_i(\bar{c_1},\bar{c})}{\*p_i(\bar{c_1})\*p_i(\bar{c})},
    $$
    which indicates that $g_k^\top = g_{k-1}^\top \*B(\*p)$.
    Now It is sufficient to prove that $g_1 = \alpha_{\pi}$.
    Straight calculation shows that
    \begin{equation*}
        g_1(i,c) = \sum_{c_1\in \+L(e_i)} \inner{\*b_{i,c_1}^\phi}{\*b_{i,c}^\phi}(\1{\pi(i) = c_1} - \*p_r(i,c_1)) = \alpha_\pi(i,c_2).
        \qedhere
    \end{equation*}
\end{proof}
In the following, we omit $(\*p)$ in $\*B(\*p)$ for brevity if there is no ambiguity. \Cref{lem:recursion_for_g} directly indicates that we can use the 2-norm of $\*B$ to upper bound that of $\*A$.
\begin{lemma}\label{lem:convert_A_to_B}
    Let $\*B \in \mathbb R^{X\times X}$ and $\alpha_\pi$ denote the matrix and the vector defined in ~\Cref{lem:recursion_for_g}.
    Then we have that for any $k\geq 1$,
    \begin{align}\label{eq:lem8}
        \sum_{\pi\in C_r}\*A^k(\pi,\pi)
        = \sum_{\pi\in C_r}f(\*p)(\pi)\alpha_\pi^\top \*B^{k-1}\beta_\pi
    \end{align}
    where $\beta_\pi(j,c_4) = \1{\pi(j) = c_4} - \*p_r(j,c_4)$, implying that $\lambda_{\max}(\*A) \leq \|\*B\|_2$.
\end{lemma}
\begin{proof}[Proof of~\Cref{lem:convert_A_to_B}]
    ~\Cref{eq:lem8} immediately follows from ~\Cref{eq:trace_of_A} and ~\Cref{lem:recursion_for_g}.
    Therefore, the maximum eigenvalue of $\*A$ can be expressed as follows.
    \begin{align*}
        \lambda_{\max}(\*A) &= \lim_{k\to \infty}\tp{\sum_{\pi\in C_r}f(\*p)(\pi)\alpha_\pi^\top \*B^{k-1}\beta_\pi}^{\frac 1k}
        \\&\leq \lim_{k\to \infty}\tp{\sum_{\pi\in C_r}f(\*p)(\pi)\norm{\alpha_\pi}_2 \norm{\*B^{k-1}\beta_\pi}_2}^{\frac 1k}
        \\&\leq \lim_{k\to \infty}\tp{\sum_{\pi\in C_r}f(\*p)(\pi)\norm{\alpha_\pi}_2 \norm{\beta_\pi}_2}^{\frac 1k} \norm{\*B}_2^{\frac{k-1}{k}}
        \\&= \|\*B\|_2.
    \end{align*}
\end{proof}
\subsubsection{Bound the transition matrix} \label{sec:bound_Jac}
Let $D_T\in \mathbb{R}^{X\times X}$ be the diagonal matrix where $D_T((i,c),(i,c)) = \*p_i(\bar{c})$.
In this section, we give an upper bound for $\norm{\*B}_2$ as $\*B$ can be represented as the product of covariance matrices of $\*p_r$, $\*p_i$ and some auxiliary diagonal matrices.

\begin{proposition}\label{prop:convert_B_to_Cov}
    Let $C_i\in \mathbb{R}^{X_{v_i}\times |\+L(e_i)|}$ denote the matrix satisfying that $C_i((j,c_1),c_2) = \1{c_1 = c_2}$ for any $i\in [d]$.
    Let $\*R \in \mathbb{R}^{X\times X}$ denote $$
    \operatorname{diag}\{C^\top_i \!{Cov}(\*p_i)C_i\}_{i\in [d]}.
    $$
    Then we have that
    $$
        \*B =  \!{Cov}(\*p_r)D_T^{-1}\*RD_T^{-1}.
    $$
\end{proposition}
\begin{proof}[Proof of~\Cref{prop:convert_B_to_Cov}]
Note that for any $i\in [d]$ and $c_1,c_2\in \+L(e_i)$,
    \begin{align*}
        \*R((i,c_1),(i,c_2)) &= C_i^\top  \!{Cov}(\*p_i)C_i(c_1,c_2)
        \\&= \*p_i(\bar{c_1},\bar{c_2})-\*p_i(\bar{c_1})\*p_i(\bar{c_2}).
    \end{align*}
    Therefore, for any $c_2\in \+L(e_i)$ and $c_4\in \+L(e_j)$, the following always holds.
    \begin{align*}
        \*B((i,c_2),(j,c_4)) &= 
        \sum_{c_3 \in \+L(e_j)}
        (\*p_r(j,c_3,i,c_2) - \*p_r(j,c_3)\*p_r(i,c_2))
        \times \frac{\*p_i(\bar{c_3},\bar{c_4})}{\*p_i(\bar{c_3})\*p_i(\bar{c_4})}
        \\&= \sum_{c_3\in \+L(e_j)}  \!{Cov}(\*p_r)((i,c_2),(j,c_3)) \tp{\frac{\*p_i(\bar{c_3},\bar{c_4})}{\*p_i(\bar{c_3})\*p_i(\bar{c_4})} - 1}
        \\&= \sum_{(k,c_3): c_3\in \+L(e_k)}  \!{Cov}(\*p_r)((i,c_2),(k,c_3))
        \frac{1}{\*p_k(\bar{c_3})}
        \*R((k,c_3),(j,c_4))
        \frac{1}{\*p_k(\bar{c_4})} .
    \end{align*}
    where the second equality comes from $\sum_{c_3 \in \+L(e_j)}
        (\*p_r(j,c_3,i,c_2) - \*p_r(j,c_3)\*p_r(i,c_2))=0$.
\end{proof}
Then we can establish~\Cref{thm:bound_Jacobian} through the above conclusions.
\begin{proof}[Proof of~\Cref{thm:bound_Jacobian}]
The Loewner Order still holds under the Congruent transformation. 
Therefore, by spectral independence of $\*p_i$, we have that
\begin{align}\label{eq:bound_R}
    \*R \preceq (1+\eta)\operatorname{diag}\{C^\top_i\Pi(\*p_i)C_i\}_{i\in [d]}
    = (1+\eta)(I - D_T).
\end{align}
Plugging in $\-{Cov}(\*p_r)\preceq (1+\eta)\Pi(\*p_r)$ and~\Cref{eq:bound_R}, we get
\begin{align}
    \nonumber \|\*B\|_2 &\leq (1+\eta)^2\lambda_{\max}(\Pi(\*p_r)D_T^{-1}(I - D_T)D_T^{-1})
    \\ \label{eq:marginal_equation} &= (1+\eta)^2\max_{(i,c): c\in \+L(e_i)} \frac{\*p_r(i,c)\*p_i(c)}{\*p_i(\bar{c})^2}.
\end{align}
Applying marginal bounds for~\Cref{cond:marginal}, we have that
\begin{align*}
    \|\*B\|_2  &\leq (1+\eta)^2\max_{(i,c): c\in \+L(e_i)} \frac{\*p_i(c)}{\*p_i(\bar{c})(\beta - 1)}
    \\&\leq (1+\eta)^2\max_{i\in[d]} \frac{\deg(v_i)-1}{(\beta - 1)^2} 
    \\&\leq \frac {1 - 2\delta}{\Delta}
\end{align*}
for some $\delta > 0$.
The last inequality follows from $\beta \geq 1 + \frac{(1 + \eta)\Delta}{\sqrt{1 - 2\delta}}$.
By~\Cref{lem:convert_A_to_B}, $\norm{(\+J f^\phi)(\*p)}^2 \leq \|\*B\|_2$.
Therefore, 
\[
    \norm{(\+J f^\phi)(\*p)}_2 \leq \sqrt{\|\*B\|_2} \leq \frac{1-\delta}{\sqrt \Delta}.
    \qedhere
\]
\end{proof}
\subsection{Strong spatial mixing via contraction}\label{sec:contraction}
\newcommand{\ra}{\rightarrow}
As \Cref{thm:bound_Jacobian} gives the upper bounds on the 2-norm of the Jacobian matrix, we now proceed to demonstrate how these bounds can be used to prove strong spatial mixing via contraction. Specifically, we will quantify the decay of correlations using the derived bounds.
Let $B_{G}(u,d)$ denote $\set{u'\in V\mid \dist_G(u,u') = d}$.

\begin{proof}[Proof of~\Cref{thm:SSM}]
Fix $e_r = (u, r)\in E$ and $r$ is the root of the tree.
Let $\tau_1$ and $\tau_2$ be two different feasible pinnings on $\Lambda \subseteq E\setminus \set{e_r}$.
We use $(T',\+L')$ to denote the edge coloring instance which is obtained by removing every $e\in \Lambda \setminus \partial_{\tau_1,\tau_2}$ from $T$ and removing $\tau_1(e)$ from the lists of the neighbours of $e$. It is easy to verify that $(T',\+L')$ is still a $\beta$-extra edge coloring instance.

Let $\ell := \min_{e\in \partial_{\tau_1,\tau_2}} \-{dist}_{T'}(e_r,e) - 1$.
It is trivial if $\ell = \infty$.
Without loss of generality, assume $\ell \geq 3$.
Let $\*p_\ell$ and $\*p_\ell'$ denote the marginal distribution over $\bigcup_{u\in B_{T'}(r,\ell)} E_{T_u'}(u)$ under the pinning $\tau_1$ and $\tau_2$ respectively.
Let $f^{\phi,i\ra i-1}$ and $f^{i\ra i-1}$ denote the concatenation of recursive function $f^\phi$ and $f$ of subtrees which is rooted at $B_{T'}(r, i-1)$ respectively.
For simplicity, let $f^{\phi,i\ra j} := f^{\phi,j+1\ra j} \circ f^{\phi,j+2\ra j+1} \circ \dots \circ f^{\phi,i\ra i-1}$ for any $i > j$ and it is the same for $f^{i\ra j}$.
We use $\*m_\ell(t)$ denote the linear combination of  $\phi(\*p_\ell)$ and $\phi(\*p_\ell')$, that is,
$\*m_\ell(t) = t\phi(\*p_\ell) + (1-t) \phi(\*p_\ell')$.
Then we have that
\begin{align}
    \nonumber \norm{f^{\phi,\ell\ra 0}(\*m_\ell(0)) - f^{\phi,\ell\ra 0}(\*m_\ell(1))}_2
    &= \norm{\int_0^1 (\+J f^{\phi,\ell\ra 0}(\*m_\ell(t))) \cdot (\phi(\*p_\ell) - \phi(\*p_\ell')) \d t}_2
    \\\nonumber &\leq \int_0^1 \norm{ (\+J f^{\phi,\ell\ra 0}(\*m_\ell(t))) \cdot (\phi(\*p_\ell) - \phi(\*p_\ell')) }_2\d t
    \\\nonumber &\leq \max_{t\in[0,1]} \norm{ (\+J f^{\phi,\ell\ra 0}(\*m_\ell(t)))}_2\norm{(\phi(\*p_\ell) - \phi(\*p_\ell')) }_2
    \\ \label{eq:contraction} &\leq \max_{t\in[0,1]} \norm{ (\+J f^{\phi,\ell\ra 0}(\*m_\ell(t)))}_2
    \max_{u\in B_{T'}(r,\ell-1)}\Delta^{\frac{\ell}{2}}\norm{\phi(\*p_u) - \phi(\*p_u')}_2
\end{align}
where $\*p_u$ and $\*p_u'$ are the marginal distributions over $E_{T_u'}(u)$ on $T_u'$ under the pinning $\tau_1$ and $\tau_2$ respectively.
As $\phi(x) = 2\sqrt{x}$, $\norm{\phi(\*p_u) - \phi(\*p_u')}_2 \leq \norm{\phi(\*p_u')}_2 + \norm{\phi(\*p_u)}_2 = 4$.
And we have that for any $t\in [0,1]$,
\begin{align}
    \nonumber \norm{ (\+J f^{\phi,\ell\ra 0}(\*m_\ell(t)))}_2 &= \norm{(\+J f^{\phi,1\ra 0}(\*m_1(t)))\dots (\+J f^{\phi,\ell\ra \ell-1}(\*m_\ell(t)))  }_2
    \\\nonumber &\leq \prod_{i=1}^{\ell}\norm{\+J f^{\phi,i\ra i-1}(\*m_i(t))}_2
    \\\nonumber &= \prod_{i=1}^{\ell}\sup_{u\in B_{T'}(r,i-1)}\norm{\+J f^{\phi}(\*m_u(t))}_2
    \\\label{eq:apply_trivial_bound} &\leq 16\Delta q \tp{\sup_{u\in B_{T'}(r,d), d\leq \ell-3}\norm{\+J f^{\phi}(\*m_u(t))}_2}^{\ell-2}
\end{align}
where $\*m_i(t) := f^{\phi,\ell\ra i}(\*m_\ell(t)) = \phi(f^{\ell\ra i}(\phi^{-1}(\*m_\ell(t))))$ and $\*m_i(t) = (\*m_u(t))_{u\in B_{T'}(r,i-1)}$.
Here \Cref{eq:apply_trivial_bound} follows from the following claim.
\begin{claim}\label{claim:trivial_bound}
    For any $t\in [0,1]$, we have that
    \begin{enumerate}
    \item $\norm{\+J f^\phi(\*m_u(t))}_2 \leq 4\sqrt{2\Delta q}$ holds for $u\in B_{T'}(r,\ell-1)$.
    \item $\norm{\+J f^\phi(\*m_u(t))}_2 \leq 2\sqrt{2\Delta q}$ holds for $u\in B_{T'}(r,\ell-2)$.
    \end{enumerate}
\end{claim}
\begin{proof}[Proof of \Cref{claim:trivial_bound}]
    Fix $u\in B_{T'}(r,\ell-1)$.
    By the concavity of $\phi$, we have that for any $\tau \in C_u$ and $t\in [0,1]$,
    $$
        \phi^{-1}(\*m_u(t))(\tau)
        = \phi^{-1}(t \phi(\*p_u(\tau)) + (1-t)\phi(\*p_u'(\tau)))
        \leq t \*p_u(\tau) + (1-t)\*p_u'(\tau);
    $$
    $$
        \phi^{-1}(\*m_u(t))(\tau)
        \geq t^2\*p_u(\tau) + (1-t)^2\*p_u'(\tau).
    $$
    Since $\ell = \min_{e\in \partial_{\tau_1,\tau_2}} \-{dist}_{T'}(e_r,e) - 1$, by \Cref{lem:marginal_bound_1} and \Cref{lem:marginal_bound_2}, $\*p_\ell$ and $\*p_\ell'$ satisfy \Cref{cond:marginal}, thus we have that
    $$
        \norm{\+J f^\phi(\*m_u(t))}_2^2 \leq \norm{\*B(\phi^{-1}(\*m_u(t)))}_1
        \leq \Delta q\times \frac{1}{\tp{\frac 12 \frac{\beta -1}{\beta -1 + \Delta}}^2}\times 2
        = 32 \Delta q.
    $$
    For any $u\in B_{T'}(r,\ell-2)$, by \Cref{lem:marginal_bound_1} and \Cref{lem:marginal_bound_2}, $\phi^{-1}(\*m_u(t))$ satisfies \Cref{cond:marginal}.
    Therefore, 
    \[
        \norm{\+J f^\phi(\*m_u(t))}_2^2 \leq \norm{\*B(\phi^{-1}(\*m_u(t)))}_1
        \leq \Delta q\times \frac{1}{\tp{\frac{\beta -1}{\beta -1 + \Delta}}^2}\times 2
        = 8 \Delta q.
        \qedhere
    \]
\end{proof}
The second equation follows from $\+J f^{\phi,i\ra i-1}(\*m_i(t)) = \operatorname{diag}\{\+J f^{\phi}(\*m_u(t))\}_{u\in B_{T'}(r,i-1)}$.
By \Cref{lem:marginal_bound_1}, \Cref{lem:marginal_bound_2} and \Cref{lem:SI-broom}, for any $u\in B_{T'}(r,d)$ and $d\leq \ell-3$, $\*m_u(t)$ satisfies~\Cref{cond:marginal} and is $(1+\eta_\Delta)$-spectrally independent.
Plugging in~\Cref{thm:bound_Jacobian} and~\Cref{eq:contraction}, we have that 
\begin{align*}
     \norm{\mu_{e_r}^{\tau_1}- \mu_{e_r}^{\tau_2}}_{\-{TV}}
    &\leq \norm{\mu_{E_T(r)}^{\tau_1}- \mu_{E_T(r)}^{\tau_2}}_{\-{TV}}
    \\&= \frac 12\norm{f^{\phi,\ell\ra 0}(\*m_\ell(0)) - f^{\phi,\ell\ra 0}(\*m_\ell(1))}_1
    \\&\leq \frac{\sqrt{q^\Delta}}2\norm{f^{\phi,\ell\ra 0}(\*m_\ell(0)) - f^{\phi,\ell\ra 0}(\*m_\ell(1))}_2
    \\&\leq 32q^{\frac{\Delta + 2}2}\Delta^2(1-\delta)^{\ell-2}.
\end{align*}

We pick $C = \max\{32q^{\frac{\Delta + 2}2}\Delta^2(1-\delta)^{-3},(1-\delta)^{-4}\}$ to finish the proof.
\end{proof}

\subsection{Worst-case scenario}\label{sec:limit}

Though \Cref{thm:bound_Jacobian} establishes an upper bound on the 2-norm of the Jacobian matrix under certain conditions, in this section, we will introduce the ``worst'' pinning of $q$-edge coloring, which is the bottleneck of our analysis. 
In fact, with potential function $\phi(x) = 2\sqrt x$ and applying 2-norm for the correlation decay step, the best bound we expected to prove is $q > \tp{\frac{3+\sqrt 5}{2}+ o(1)}\Delta\approx 2.618\Delta$. However what we can only prove strong spatial mixing for instances of $(1+o(1))\Delta$-extra edge colorings, as we currently lack a better upper bound for $\*R$.
Note that the pinning that maximizes~\cref{eq:marginal_equation} is the same as the pinning in~\Cref{thm:worst_pinning} and~\cref{eq:marginal_equation} can indeed achieve the upper bound $\frac{\Delta-1}{(q-2\Delta+2)^2}$ under this worst pinning.

Before showing the worst-case scenario, we introduce the following technical lemma to calculate the eigenvalues of some simple matrices.
\begin{lemma}\label{lem:calc_eigenvalues}
    Given constant $k_1,k_2\neq 0$, 
    the eigenvalues of $k_1\*1\*1^\top +k_2 \!{Id}_n$ are either $k_2$ or $nk_1 + k_2$ where $\!{Id}_n$ is the identity matrix in $\mathbb{R}^{n\times n}$.
\end{lemma}
\begin{proof}[Proof of~\Cref{lem:calc_eigenvalues}]
    The eigenvalues $\lambda$ and the corresponding eigenvectors $\*v$ satisfies that 
    \begin{align}\label{eq:eigen_equation}
        (k_1 \*1\*1^\top + k_2\!{Id})\*v = \lambda \*v
        \implies k_1\langle\*1,\*v\rangle \*1 = (\lambda - k_2)\*v.
    \end{align}
    Therefore, there are only two cases for~\Cref{eq:eigen_equation}.
    $$
    \begin{cases}
        \lambda = k_2 \land \*1\perp \*v
        \\\*v \parallel \*1
    \end{cases}.
    $$
    Plugging in $\*v = \*1$ and~\Cref{eq:eigen_equation}, we get
    \[
        k_1\langle\*1,\*1\rangle = \lambda - k_2 \implies \lambda = k_2 + nk_1.
        \qedhere
    \]
\end{proof}
Now we show the worst pinning and corresponding norm value of $\*B$ and $(\+J f^\phi)(\*p)$.
\begin{theorem}\label{thm:worst_pinning}
    Under some specific pinning, the norm of $\*B$ satisfies lower bound that
    $$
    \|(\+J f^\phi)(\*p)\|_2^2 = \|\*B\|_2 = \frac{\Delta-1}{(q-\Delta)(q-2\Delta+2)}.
    $$
    Then $\|\*B\|_2 < \frac 1\Delta$ implies that $q > \tp{\frac{3+\sqrt 5}{2}+o(1)}\Delta$.
\end{theorem}
\begin{proof}[Proof of~\Cref{thm:worst_pinning}]
    Consider the following instance of edge coloring $(G,\+L)$ generated from a $q$-coloring instance by pinning:
    \begin{enumerate}
        \item $\deg(r) = 1$, $\+L(e_1) = \{\Delta,...,q\}$, that is, in original instance $r$ has $\Delta$ children and we pin $e_2,...,e_{\Delta}$ with $1,2,...,\Delta-1$.
        \item $\deg(v_1) = \Delta$ and for any $u\in N(v_1)$ and $u\neq r$, $\+L(\{u,v_1\}) = \{\Delta,...,q\}$, that is, we assume $u$ has $\Delta-1$ children and we pin the children of $u$ with color $1,2,...,\Delta-1$.
    \end{enumerate}
    Then by~\Cref{prop:convert_B_to_Cov}$, \*B$ is equivalent to a matrix in $\mathbb{R}^{(q-\Delta+1)\times (q-\Delta+1)}$ and holds for the following equation
    $$
        \*B = -\frac{\Delta-1}{(q-\Delta+1)(q-\Delta)(q-2\Delta+2)}\tp{\*1\*1^\top - \!{Id}} + \frac{\Delta-1}{(q-\Delta+1)(q-2\Delta+1)}\!{Id}.
    $$
    $\*B$ is a symmetrical matrix, thus the singular values equals to the eigenvalues. 
    Applying~\Cref{lem:calc_eigenvalues}, we have that
    $$
        \|\*B\|_2 = \lambda_{\max}(\*B) = \frac{\Delta-1}{(q-\Delta)(q-2\Delta+2)}.
    $$
    For $(\+J f^\phi)(\*p)$, since $\abs{C_r} = |\+L(e_1)| = q-\Delta+1$, $(\+J f^\phi)(\*p)(\+J f^\phi)^\top(\*p) = \*A \in \mathbb{R}^{(q-\Delta+1)\times (q-\Delta+1)}$.
    Now it is sufficient to show that $\*A = \*B$.
    By~\Cref{eq:equation_A}, we have that
    \begin{align*}
        \*A(c_1,c_2) &= \frac 1{q-\Delta+1}\sum_{c_3,c_4\in \+L(e_1)}\langle \*b^\phi_{1,c_3}, \*b^\phi_{1,c_4}\rangle\tp{\1{c_1=c_3} - \frac 1{q-\Delta+1}}\tp{\1{c_2 = c_4} - \frac 1{q-\Delta+1}}
        \\&= \begin{cases}
            \frac{\Delta-1}{(q-\Delta+1)(q-2\Delta+2)} &, c_1 = c_2
            \\\frac{1-\Delta}{(q-\Delta+1)(q-\Delta)(q-2\Delta+2)} &, c_1\neq c_2
        \end{cases}
        \\&= \*B(c_1,c_2).
    \end{align*}
\end{proof}
\Cref{thm:worst_pinning} indicates the limitations of our current analysis by providing a lower bound on the norm of the Jacobian matrix. This indicates the best bound one can expect to use $2$-norm and the potential function $\phi(x) = 2\sqrt{x}$ is $q\approx 2.618\Delta$.


%% file: weighted.tex
\section{Covariance matrix on brooms}\label{sec:covariance}
We will give an upper bound for the covariance matrix defined in \Cref{def:covariance} in this section.
\begin{lemma}\label{lem:SI-broom}
    Given a tree $T$ rooted with $r$ of depth $\ell+1$ and weight functions $\*p_v:E_{T_v}(v)\rightarrow \bb{R}_{\geq 0}$ for any $v\in B_T(r,\ell)$. If $\ell\geq 2$ and $\beta \geq \Delta + 50$, then $\*p_r\defeq  f^{\ell\rightarrow 0}(\tp{\*p_v}_{v\in B(r,\ell)})$ is $(1+\eta_\Delta)$-spectrally independent, i.e.
    \[
    \!{Cov}(\*p_r)\mle (1+\eta_{\Delta})\Pi(\*p_r).
    \]
    where $\eta_{\Delta}$ is $O\tp{\frac{\log^2 \Delta}{\Delta}}$.
\end{lemma}
\subsection{Weighted edge coloring instance}
\label{sec:weight-edge-coloring}
Given a tree $T$ and associated $\beta$-extra color lists $\+L$ for all edges in $T$, an weighted edge coloring instance is a distribution over proper list colorings $(T, \+L)$.
There are non-intersecting connected subgraphs $\set{K_i}_{i\in [N]}$ in $T$
called \emph{boundaries} and associated \emph{weight functions} $\set{w_i}_{i\in [N]}$
such that $w_i : C_{v_i} \to \mathbb R_{\ge 0}$.
\begin{definition}[Distribution of weighted list edge coloring]\label{def:weighted-coloring}    
    Given a list edge coloring instance $(T, \+L)$,
    boundaries $\set{K_i}_{i\in [N]}$
    and weighted functions $\set{w_i}_{i\in [N]}$
    as above, define the associated weighted edge coloring instance
    be the distribution $\mu$ on $\Omega$ such that
    \begin{align*}
        \mu(\sigma) \defeq \frac1Z \prod_{i\in [N]} w_i(\sigma|_{K_i}),
    \end{align*}
    where $Z\defeq \sum_{\sigma\in \Omega}\prod_{i\in [N]} w_i(\sigma|_{K_i})$
    is the partition function.
\end{definition}

Notice that, if we pick an $\omega_i\in\Omega_{K_i}$ for each $i\in [N]$,
then
\[\mu^{\omega_1\cup\cdots\cup \omega_N}(\sigma)
\propto
\begin{cases}
 \frac{1}{Z}   \prod_{i\in [N]} w_i(\sigma|_{K_i}),  &\forall i\in [N], \sigma|_{K_i} = \omega_i,
 \\ 0, & \text{otherwise},
\end{cases}
\]
which means $\mu$ is a distribution over uniform edge colorings conditioned on a 
pinning on all $K_i$. This leads to the following lemma.
\begin{lemma}\label{lem:weighted-decomposition-to-unweighted}
    We have
    \[
        \mu = \sum_{\omega_1\in \Omega_{K_1},\cdots, \omega_N\in \Omega_{K_N}}
        \mu_{K_1\cup\cdots\cup K_N}(\omega_1\cup\cdots\cup\omega_{N}) \mu^{\omega_1\cup\cdots\cup\omega_N}.
    \]
    That is, the distribution on weighted coloring instance in \Cref{def:weighted-coloring}
    can be expressed as a mixture of uniform distributions over edge colorings.
\end{lemma}
\begin{proof}
    This is just an application of the total probability formula.
\end{proof}

With \Cref{lem:weighted-decomposition-to-unweighted}, we can prove
weighted edge colorings inherit the marginal bound \Cref{lem:claw-marginal-generalized}
on non-boundary edges.
\begin{lemma}[Marginal upper bound - weighted version]\label{lem:marginal-upper-weighted}
    Consider a weighted edge coloring instance on a tree $T=(V, E)$
    as defined in \Cref{def:weighted-coloring},
    whose distribution is denoted by $\mu$, and a pinning $\xi\in \Omega_{S}$ for some subset $S\subset E$.
    Then for any $i\in V$, $F\subseteq E_i\setminus \tp{\bigcup_{i\in [N]}K_i\cup S}$ and color $a$:
    denoting $\beta\defeq\min_{e\in F}\set{\abs{\+L^\xi(e)}-\deg(e)}$, we have
    \[\mu^\xi(a\in c(F)) \le \frac{|F|}{\beta - 1 + |F|}.\]
\end{lemma}
\begin{proof}
    We write $\mu^\xi$ into a mixture of uniform distributions on edge colorings
    by \Cref{lem:weighted-decomposition-to-unweighted}:
    \begin{align*}
        \mu^\xi = \sum_{\omega_1\in \Omega_{K_1},\cdots, \omega_N\in \Omega_{K_N}}
        \mu^\xi_{K_1\cup\cdots\cup K_N}(\omega_1\cup\cdots\cup\omega_{N}) \mu^{\xi\cup\omega_1\cup\cdots\cup\omega_N}.
    \end{align*}
    It suffices to prove the bound for any $\mu^{\xi\cup\omega_1\cup\cdots\cup\omega_N}$.
    Notice that $\mu^{\xi\cup\omega_1\cup\cdots\cup\omega_N}$ are uniform distributions 
    on proper colorings defined on $T\setminus\tp{\bigcup_{i\in [N]}K_i\cup S}$, with $\beta$-extra colors 
    on all edges in $F$,
    so \Cref{lem:claw-marginal-generalized} applies and proves the lemma.
\end{proof}

We also need the following lower bound in the matrix trickle-down.
\begin{lemma}\label{lem:marginal-lower-weighted}
    Consider a weighted edge coloring instance on a tree $T$ with $\beta$-extra
    color lists $\+L$.
    If $K=E(v)$ for some vertex $v$ in $T$ and none of the neighbours of $K$ is in the boundary,
    $\tau\in \Omega_S$ is a pinning on some subset $S \subset E$, and the number of 
    unpinned edges in $K$ is $k$. Then for any unpinned $e\in K$ and $c\in \+L^\tau(e)$,
    \[
    \mu^{\tau}_e(c) \ge \frac{(\beta-1)^2}{(\beta+k-2)(\beta+\Delta-2)}\frac{1}{\ell_e^\tau-k+1}.
    \]
    where $\ell^\tau_e\defeq|\+L^\tau(e)|$.
\end{lemma}
\begin{proof}
    We assume $e = \set{u, v}$, where $K=E(v)$ and $L\defeq E(v)$.
    By \Cref{lem:weighted-decomposition-to-unweighted},
    it suffices to prove the bound for any $\mu^{\sigma\cup\omega_1\cup\cdots\cup\omega_N}$.
    Denoting $\xi$ as the shortcut for $\sigma\cup\omega_1\cup\cdots\cup\omega_N$,
    we can do a further decomposition:
    \begin{align*}
        \mu^{\xi}_e(c) =
        \sum_{\substack{\sigma_1\in\Omega^{\xi            }_{K\setminus\set e}\\ c\notin \sigma_1}}
        \mu^{\xi                        }_{K\setminus\set e}(\sigma_1)\cdot
        \sum_{\substack{\sigma_2\in\Omega^{\xi\cup\sigma_1}_{L\setminus\set e}\\ c\notin \sigma_2}}
        \mu^{\xi\cup\sigma_1            }_{L\setminus\set e}(\sigma_2)\cdot
        \mu^{\xi\cup\sigma_1\cup\sigma_2}_e(c).
    \end{align*}
    By assumption, $K$ and $L$ are not in the boundary, then \Cref{lem:marginal-upper-weighted}
    gives
    \begin{align*}
       &\sum_{\substack{\sigma_1\in\Omega^{\xi            }_{K\setminus\set e}\\ c\notin \sigma_1}}
        \mu^{\xi                        }_{K\setminus\set e}(\sigma_1)
        \ge \frac{\beta-1}{\beta + k - 2}
     \\&\sum_{\substack{\sigma_2\in\Omega^{\xi\cup\sigma_1}_{L\setminus\set e}\\ c\notin \sigma_2}}
        \mu^{\xi\cup\sigma_1            }_{L\setminus\set e}(\sigma_2)
        \ge \frac{\beta-1}{\beta + \Delta - 2}.
    \end{align*}
    Moreover, since $\mu^{\xi\cup\sigma_1\cup\sigma_2}_e$ is the uniform distribution
    over $\+L^{\xi\cup\sigma_1\cup\sigma_2}(e)$, whose size is at most $\ell_e^{\tau}-k+1$, it is
    lower bounded by $1/(\ell_e^{\tau}-k+1)$.

    Combining the three bounds, we have
    \[
        \mu^{\tau}_e(c) \ge \frac{(\beta-1)^2}{(\beta+k-2)(\beta+\Delta-2)}\frac{1}{\ell_e^\tau-k+1}.
        \qedhere
    \]
\end{proof}

After introducing the concept of weighted edge coloring instances, we now turn our attention to the marginal distribution on a broom. For simplicity we run the matrix trickle down theorem on one broom $K$ in $T$,
that does not intersect with the boundaries
(actually on the quotient simplicial complex on $K$).
Then it is necessary to look at the marginal distribution of $\mu$
on $K$ under some pinning $\xi \in \Omega_F$ on a subset of edges $F$.
We have
\begin{align}\label{eq:broom-marginal-1}
\mu^\xi_K(\tau)
&\propto \sum_{\substack{\sigma:\sigma|_K=\tau \\ \sigma|_F=\xi}}
 \prod_{i\in [N]} w_i(\sigma|_{K_i}).
\end{align}

The following lemma demonstrates the relation between the weighted coloring instance and the tree recursion.
\begin{lemma}\label{lem:p=mu}
    Assume the tree $T$ with root $r$ is of depth $\ell+1$. Let the boundaries be all brooms on $\ell+1$-th level.
    For any $v\in B(r,\ell)$, the weight function on $E_{T_v}(v)$ is just $\*p_v$, and $\*p=(\*p_v)_{v\in B(r,\ell)}$. Choose $K=E(r)$. Then for the simplicial complex defined as above, we have  \[\*p_r\defeq  f^{\ell\rightarrow 0}(\*p)= \mu_K .\]
\end{lemma}
\begin{proof} 
    For simplicity, we assume that all leaves in $T$ are at the same level. It can be generalized to any tree of depth $\ell+1$. We write the entries of $\*p_r$ explicitly:
    \begin{align*} 
    \*p_r(\pi) &\propto \prod_{v^1\in N(r)} \sum_{\substack{\sigma^1 \in  C_{v^1}\\\pi((r,v^1))\notin \sigma^1}} \*p_{v^1}(\sigma^1)\\
    &\propto \prod_{v^1\in N(r)}
    \sum_{\substack{\sigma^1 \in C_{v^1}\\\pi((r,v^1))\notin \sigma^1}} \cdots
    \prod_{v^\ell\in N(v^{\ell-1})}
    \sum_{\substack{\sigma^\ell \in C_{v^{\ell-1}}\\\sigma^{\ell-1}((v^{\ell-1},v^\ell))\notin \sigma^\ell} \*p_{v^\ell}}(\sigma^\ell)\\
    & = \sum_{\sigma|_K=\pi}\prod_{v \in B(v,\ell)} \*p_v(\sigma|_{E_{T_v}(v)}).
    \end{align*}
    Therefore, $\*p_r = \mu_K$.
\end{proof}

Suppose $K = \set{e_i}_{i\in [d]}$.
Then $T\setminus K$ is composed of $d$ disconnected subgraphs denoted by $T_i$
such that $T_i$ is adjacent to $K$ in $T$ (we define $T_i = \emptyset$
if $e_i$ is pendant).
Moreover, we define $\+K_i\defeq\set{K_j\mid K_j\subseteq T_i}$.
Since both $\set{T_i}$ and $\set{\+K_i}$ contain disjoint subgraphs and 
the $\sigma$ in the summation in \cref{eq:broom-marginal-1}
is determined by partial colorings $\sigma|_{T_i}$,
we can define $\+S^\xi_{i, c}\subseteq \Omega^\xi_{T_i}$ by the proper colorings 
on $T_i$ such that is compatible with $e_i$ being colored $c$.
\begin{align*}
\mu^\xi_K(\tau)
&\propto
  \sum_{\sigma_1\in \+S^\xi_{1, \tau(e_1)}}
  \sum_{\sigma_2\in \+S^\xi_{2, \tau(e_2)}}\dots
  \sum_{\sigma_N\in \+S^\xi_{N, \tau(e_N)}}
  \prod_{i\in [d]} \prod_{j : K_j\in \+K_i} w_j(\sigma_i|_{K_j})\\
&= \prod_{i\in [d]} \sum_{\sigma_i\in \+S^\xi_{i, \tau(e_i)}}
     \prod_{j : K_j\in \+K_i} w_j(\sigma_i|_{K_i}).
\end{align*}
Define 
\[
p^\xi_{e_i,c}=\sum_{\sigma_i\in \+S^\xi_{i, c}}
     \prod_{j : K_j\in \+K_i} w_j(\sigma_i|_{K_i}).
\]
Then we have $\mu^\xi_K(\tau)\propto\prod_{i\in [d]} p^\xi_{e_i,\tau(e_i)}$.
This proves the following lemma.
\begin{lemma}[Marginal distribution on a broom]
    Consider a list edge coloring instance $(T, \+L)$,
    boundaries $\set{K_i}_{i\in [N]}$
    weighted functions $\set{w_i}_{i\in [N]}$,
    and a pinning $\xi\in \Omega_{F}$ for some subset $F\subset E$.
    Then for a broom $K$ disjoint from boundaries
    there exists constants $p^\xi_{e, c}$ for $e\in K$, $c\in \+L^\xi(e)$ such that
    \[
    \mu^\xi_K(\tau)\propto\prod_{i\in [d]} p^\xi_{e_i,\tau(e_i)}.
    \]
\end{lemma}

Let $q^\xi_e=\sum_{c\in \+L^\xi(e)} p^\xi_{e,c}$
and $q^\xi_{e,f}=\sum_{c\in \+L^\xi(e)\cap \+L^\xi(f)} p^\xi_{e,c}p^\xi_{f,c}$.

Next we present some bounds on the quantities $p^\xi_{e, c}$ and $q^\xi_\cdot$
\begin{lemma}\label{lem:marginal-bound-weighted-1}
For any $e_i\in K, \xi\in \Omega_F$ for some subset $F \subset E$,
\[
\frac{p^\xi_{e_i,c}}{q^\xi_{e_i}}\le \frac{1}{\ell^\xi_{e_i}} .
\]
\end{lemma}
\begin{proof}
    Denote the subtree induced by $T_i$ and $e_i$ by $\tilde T_i$.
    We consider the weighted coloring instance on $(\tilde T_i, \+L)$ with boundaries
    $\+K_i$ and weighted functions $\set{w_j}_{K_j\in \+K_i}$
    and denote the associated distribution by $\nu$.
    Then by \Cref{def:weighted-coloring},
    \[
    \nu^\xi_{e}(c) = \frac{p^\xi_{e_i, c}}{q^\xi_{e_i, c}}.
    \]
    Then
    \[
    \frac{p^\xi_{e_i,c}}{q^\xi_{e_i}}
    =
    \sum_{\sigma\in C_{T_i, \+L}} \nu^{\sigma\cup\xi}_e(c)\nu^\xi_{T_i}(\sigma).
    \]
By \Cref{lem:marginal-upper-weighted},
$\nu^{\sigma\cup\xi}_e(c)\leq \frac{1}{\beta}$.
On the other hand, $\sum_{\sigma\in C_{T_i, \+L}} \nu^\xi_{T_i}(\sigma)=1$.
So $\frac{p_{e,c}}{q_e}\leq \frac{1}{\beta}$.
\end{proof}
\begin{lemma}\label{lem:marginal-bound-weighted-2}
For $e,f\in K$ and $\xi\in \Omega_F$ for some subgraph $F$,
    \[
    q^\xi_{e,f}\leq \frac{q^\xi_e q^\xi_f}{\beta} .
    \]
\end{lemma}
\begin{proof}
    Direct calculation yields
\begin{align*}
q^\xi_{e,f}&=\sum_{c\in \+L(e)\cap \+L(f)} p^\xi_{e,c}p^\xi_{f,c}\\
&\leq \tp{\sum_{c\in \+L(e)} \tp{p^\xi_{e,c}}^2}^{1/2} \tp{\sum_{c\in \+L(f)} (p^\xi_{f,c})^2}^{1/2}\\
&\leq \frac{\sqrt{q^\xi_e q^\xi_f}}{\beta}\tp{\sum_{c\in \+L^\xi(e)} p^\xi_{e,c}}^{1/2} \tp{\sum_{c\in \+L^\xi(f)} p^\xi_{f,c}}^{1/2}\\
&=\frac{q^\xi_e q^\xi_f}{\beta}.
\end{align*}
\end{proof}

\subsection{Simplicial complexes}
First we introduce simplicial complexes to encode the edge coloring instance. Given a universe $U$, a \emph{simplicial complex} $\+C\subseteq 2^U$ is a collection of subsets of $U$ that is downward close, which means that if $\sigma\in \+C$ and $\sigma'\subseteq\sigma$, then $\sigma'\in \+C$. Every element $\sigma\in \+C$ is called a \emph{face}, and a face that is not a proper subset of any other face is called a \emph{maximal face} or a \emph{facet}. The dimension of a face $\sigma$ is $\!{dim}(\sigma)\defeq \abs{\sigma}$, namely the size of $\sigma$. For every $k\ge 0$, let $\+C_k\defeq \set{\sigma\in\+C\cmid \abs{\sigma}=k}$ be the set of faces of dimension $k$. Specifically, $\+C_0=\set{\emptyset}$. The dimension of $\+C$ is the maximum dimension of faces in $\+C$. 

Besides, we say $\+C$ is a \emph{pure} $n$-dimensional simplicial complex if all maximal faces in $\+C$ are of dimension $n$. In this work we only deal with pure simplicial complexes. In a pure simplicial complexe, we define the co-dimension of a face $\sigma$ as $\!{codim}(\sigma)\defeq n-\!{dim}(\sigma)$.

Let $\pi_n$ be a distribution over the maximal faces $\+C_n$. We use the pair $(\+C,\pi_n)$ to denote a \emph{weighted simplicial complex} where for each $1\le k < n$, the distribution $\pi_n$ induces a distribution $\pi_k$ over $\+C_k$. Formally, for every $1\le k<n$ and every $\sigma'\in \+C_k$, $\pi_k(\sigma')$ is proportional to the sum of weights of maximal faces containing $\sigma$. Formally,
\begin{align*}
    \pi_k(\sigma') \defeq \frac{1}{\binom n k}\sum_{\sigma\in \+C_n\cmid\sigma' \subset \sigma} \pi_{n}(\sigma).
\end{align*}
It can be easily verified that $\pi_k$ is a distribution on $\+C_k$. Sometimes, we omit the subscript when $k=1$, i.e., we write $\pi$ for $\pi_1$.

Also we define the simplicial complexes generated by pinning a face in $\+C$. For a face $\tau \in \+C$ of dimension $k$, we define its \emph{link} as
\begin{align*}
    \+C_\tau=\set{\sigma\setminus \tau \cmid \sigma \in \+C \land \tau \subseteq \sigma }.
\end{align*}
Obviously, $\+C_\tau$ is a pure $(n-k)$-dimensional simplicial complex. Similarly, for every $1\le j \le n-k$, $\+C_{\tau,j}$ is the set of faces in $\+C_\tau$ of dimension $j$. We also use $\pi_{\tau,j}$ to denote the \emph{marginal distribution} on $\+C_{\tau,j}$. Formally, for every $\sigma\in \+C_{\tau,j}$,
\begin{align*}
    \pi_{\tau,j}(\sigma)\defeq \Pr[\alpha \sim \pi_{k+j}]{\alpha=\tau \cup \sigma \cmid \tau\subseteq \alpha}=\frac{\pi_{k+j}(\tau \cup \sigma)}{\binom{k+j}{k} \cdot \pi_{k}(\tau)}.
\end{align*}
We also drop the subscript when $j=1$, i.e., we write $\pi_\tau$ for $\pi_{\tau,1}$. We define a random walk $P_\tau$ with stationary distribution $\pi_\tau$ as
\[
P_\tau(x,y)=\frac{\pi_{\tau,2}(\set{x,y})}{2\pi_\tau(x)}.
\]


\subsection{Matrix trickle down on a broom}

Now we define the corresponding simplicial complex to the weighted edge coloring instance defined in \Cref{sec:weight-edge-coloring}.
Recall that we are dealing with a weighted edge coloring instance on a tree $T=(V, E)$
with $\beta$-extra color lists $\+L$. $K = E_i$ for some $i\in V$ and $|K| = d$.
We are going to show that the distribution of weighted colorings on $K$ can be
represented by a weighted simplicial complex.

Since any proper edge coloring  $\sigma:E\rightarrow [q]$ could be regarded as a set of pairs of edge and color, namely $\set{(e,c)\in E\times [q]\cmid \sigma(e)=c}$, the weighted edge-coloring instance restricted on $K$ can be naturally represented as a pure weighted simplicial complex $(\+C,\pi_{d})$ where $\+C$ consists of all proper partial colorings in $(K,\+L|_K)$ and $\pi_{d}=\mu_K$.

Before introducing the matrix trickle-down theorem, we define notations for matrices related to $\pi_\tau$. Define $\Pi_\tau \in \bb R^{\+C_1\times\+C_1}$ as $\Pi_\tau \defeq \operatorname{diag}(\pi_\tau)$ supported on $\+C_{\tau,1}\times \+C_{\tau,1}$,
and $\pi_\tau\defeq[\pi_\tau(x)]_{x\in \+C_1}$ be a vector supported on $C_{\tau, 1}$

For convenience, define the pseudo inverse $\Pi_\tau^{-1} \in \bb R^{\+C_1\times\+C_1}$ of $\Pi_\tau$ as $\Pi_\tau^{-1}(x,x)=\pi_\tau(x)^{-1}$ for $x\in \+C_{\tau,1}$ and $0$ otherwise. Similarly, the pseudo inverse square root $\Pi_\tau^{-1/2} \in \bb R^{\+C_1\times\+C_1}$is defined as $\Pi_\tau^{-1/2}(x,x)=\pi_\tau(x)^{-1/2}$ for $x\in \+C_{\tau,1}$ and $0$ otherwise. When $\tau=\emptyset$, we omit the subscript.

Recalling \Cref{def:covariance} and \Cref{def:diag_mean}, 
the following lemma relates the covariance of $\mu_K$ and the matrices defined in this section.
\begin{proposition}\label{lem:mtd-to-si}
    \[
    \Pi(\mu_K)=d\Pi
    \]
    and
    \[
    \!{Cov}(\mu_K) =d \tp{(d-1) \tp{\Pi P -\frac{d}{d-1}\pi \pi^\top} + \Pi}
    \]
    in the sense of padding with zeros.
\end{proposition}
The proof is by direct calculation.
We apply the matrix trickle-down theorem on $(\+C,\pi_d)$ to prove the following lemma. The main idea of the proof is almost the same as that in \cite{WZZ24} except for substituting the number of feasible colors to the weights of feasible colors w.r.t $\set{p_{e,c}}$ after pinning. And the construction of matrix upper bound is simpler since the line graph of $K$ is a clique. We include the details in \Cref{appendix-mtd}.
\begin{lemma}\label{lem:PiP-bound}
    If $K$ is not adjacent to boundaries, i.e. $\set{K_i}$ and $\beta \geq \Delta + 50$, then the simplicial complex $(\+C,\pi_d)$ defined as above satisfying
    \[
    \Pi P -\frac{d}{d-1}\pi \pi^\top \mle \frac{\eta_{\Delta}}{d-1} \Pi ,
    \]
    where $\eta_{\Delta}$ is $\+O\tp{\frac{\log^2 \Delta}{\Delta}}$.
\end{lemma}
Then \Cref{lem:SI-broom} is derived directly from \Cref{lem:PiP-bound} and \Cref{lem:mtd-to-si}.

%% file: wsm.tex
\section{Tight bound for weak spatial mixing}\label{sec:wsm}
In this section, we prove the tight bound for weak spatial mixing  of $q$-edge coloring instance on trees.
Note that the strong spatial mixing of $q$-edge coloring instance is a special case of spatial mixing of list instance (for list coloring instance, weak spatial mixing is equivalent to strong spatial mixing).
Therefore the theorem stated in this section is a weak version of spatial mixing conclusions and thus we can give a tight bound for trees. 

The main theorem for weak spatial mixing on trees is as follows.
\begin{theorem}\label{thm:tight_WSM}
    Given a tree $T = (V,E)$ with $n$ vertices, $m$ edges and maximum degree $\Delta$.
    Suppose that instance $(T,\+L)$ is a $q$-edge coloring instance (i.e. for any $e\in E$, $\+L(e) = [q]$). 
    Then we have that
    \begin{enumerate}
        \item If $q \geq 2\Delta - 1$, the edge coloring instance satisfies weak spatial mixing with rate $1 - \frac{1 - \eps}{2\Delta - (1+\eps)}$, where $\eps = \max\{\frac{\Delta - 1}{\Delta}, \frac {e - 1}e\}$.
        \item If $q \leq 2\Delta - 2$, there exists an instance that does not satisfy weak spatial mixing.
    \end{enumerate}
\end{theorem}
Consider the following example which simply shows that hardness of weak spatial mixing:
\begin{example}\label{eg:hardness_WSM}
    Consider a $(\Delta - 1)$-regular tree with depth $d$ (the depth of the root $r$ is $0$) and $d$ is an even number. Let $\Lambda$ be the edges between vertices of depth $d$ and $d-1$. Let $\tau_1$ be the pinning over $\Lambda$ which only uses $1,2,\dots,\Delta-1$ and $\tau_2$ be the pinning over $\Lambda$ which only uses $\Delta,\dots,2\Delta - 2$. It is easy to verify that for any $\{u,v\}\in E\setminus \Lambda$ and $\-{dep}(u) = \-{dep}(v) + 1$,
    $$
        \forall \sigma \in \Omega^{\tau_1}, , \sigma(\{u,v\}) \in 
        \begin{cases}
            \{1,\dots,\Delta-1\} &, 2\ |\ \-{dep}(u)
            \\\{\Delta,\dots,2\Delta-2\} &, 2 \nmid \-{dep}(u)
        \end{cases}
    $$
    $$
        \forall \sigma \in \Omega^{\tau_2}, \sigma(\{u,v\}) \in 
        \begin{cases}
            \{\Delta,\dots,2\Delta-2\} &, 2\ |\ \-{dep}(u)
            \\\{1,\dots,\Delta-1\} &, 2 \nmid \-{dep}(u)
        \end{cases}
    $$
    where $\-{dep}(u)$ is the depth of $u$. Therefore, for any $e\in E\setminus \Lambda$, we have that
    $$
        \|\mu_e^\sigma - \mu_e^\tau\|_{\-{TV}} = 1
    $$
    which demonstrates that the weak spatial mixing does not hold.
\end{example}
The proof scheme is also using the idea of correlation decay and we use the uniform distribution as bridge to prove the weak spatial mixing property. And we use another recursion which is different from that in strong spatial mixing.
Instead of considering a broom of edges, we specify the marginal probability of a pendant edge and then generalize to every edge.
Since the lists of feasible colors are clear, we use $\Pr[T]{\cdot}$ to denote $\Pr[T,\+L]{\cdot}$ for simplicity.
\begin{lemma}[One-step contraction]\label{lem:WSM_contraction}
Suppose $(T = (V,E), \+L)$ is a $q$-edge coloring instance, where $T$ is a tree with a pendant edge $e = \{r',r\}$ on its root $r$ (that is, $\deg(r')=1$) and $\tau$ is the pinning over a set of edges $\Lambda$ whose edges are incident to leaf vertices.  
If for any $e_i = \{v_i,r\}\in E$, the subtree $T_i$ with pendant edge $e_i$ satisfies that 
$$
    \forall c\in[q], \abs{\Pr[T_i\cup \{e_i\}]{c(e_i) = c\mid c(\Lambda) = \tau} - \frac 1q} \leq \delta
$$
where $\delta < \frac 1q$ is a universal constant, then we have that
$$
    \forall c\in [q], \abs{\Pr[T]{c(e) = c\mid c(\Lambda) = \tau} - \frac 1q} \leq \frac{2\Delta - 2}{q(1 - \delta\abs{q-2\Delta + 2})} \delta.
$$
\end{lemma}
\begin{proof}[Proof of \Cref{lem:WSM_contraction}]
    Assume that there are $d$ children of $r$.
    Let $P_{r,c}$ denote $\Pr[T]{c(e) = c\mid c(\Lambda) = \tau}$ and $P_{i,c}$ denote $\Pr[T_i\cup\{e_i\}]{c(e_i) = c\mid c(\Lambda) = \tau}$.
    Then we have the following recursion for any $c\in [q]$, 
    $$
        P_{r,c} = \frac{\displaystyle \sum_{A\in C_r, c\notin A}\prod_{i=1}^d P_{i,A_i}}{(q-d) \displaystyle \sum_{A\in C_r}\prod_{i=1}^d P_{i,A_i}}.
    $$
    $\delta < \frac 1q$ implies that $|C_r| = q^{\underline{d}}$. Therefore, the following should be true for any $c\in[q]$,
    \begin{align}
        \nonumber \frac 1{P_{r,c}} &= q - d + \sum_{i=1}^d P_{i,c}\frac{\displaystyle (q-d)\sum_{A\in C_r, A_i = c}\prod_{j\neq i}P_{j,A_j}}{\displaystyle \sum_{A\in C_r, c\notin A}\prod_{j=1}^d P_{j,A_j}}
        \\\nonumber &= q - d + \sum_{i=1}^d P_{i,c}\frac{\displaystyle \sum_{A\in C_r, c\notin A}\prod_{j\neq i}P_{j,A_j}}{\displaystyle \sum_{A\in C_r, c\notin A}P_{i,A_i}\prod_{j\neq i} P_{j,A_j}}
        \\\nonumber &\leq q - d + \sum_{i=1}^d (\frac 1q + \delta)\frac{\displaystyle \sum_{A\in C_r, c\notin A}\prod_{j\neq i}P_{j,A_j}}{\displaystyle (\frac 1q - \delta)\sum_{A\in C_r, c\notin A}\prod_{j\neq i} P_{j,A_j}}
        \\ \label{eq:left_side_WSM_contraction} &= q + \frac{2dq\delta}{1-q\delta}.
    \end{align}
    In the same way, we can show that $\frac 1{P_{r,c}} \geq q - \frac{2dq\delta}{1 + q\delta}$. 
    Combine this inequality and \Cref{eq:left_side_WSM_contraction}, we get
    $$
        \frac{-2d\delta}{q-\delta q(q-2d)}\leq P_{r,c} - \frac 1q \leq \frac{2d\delta}{q+\delta q(q-2d)}
        \implies |P_{r,c} - \frac 1q| \leq \frac{2d}{q(1 - \delta\abs{q-2d})} \delta.
    $$
    When $q > 2d$, $\frac{2d}{q(1 - \delta(q-2d))} \delta$ is monotone increasing with respect to $d$. Then $d\leq \Delta - 1$ implies that
    \[
        \abs{P_{r,c} - \frac 1q} \leq \frac{2\Delta - 2}{q(1 - \delta\abs{q-2\Delta + 2})} \delta.
        \qedhere
    \]
\end{proof}
\Cref{lem:marginal-bound-gkm} shows that the $q$-edge coloring instance admits the marginal lower bound, which is a start point of recursive contraction.
Now we can prove \Cref{thm:tight_WSM}.
\begin{proof}[Proof of \Cref{thm:tight_WSM}]
    We prove weak spatial mixing for pendant edges first.
    Let $d$ denote $\min_{e'\in \Lambda} \-{dist}_{T}(e',e)$.
    \Cref{lem:marginal-bound-gkm} implies that if $d = 2$, 
    $$
        \frac{1}{eq}\leq \Pr[T]{c(e) = c\mid c(\Lambda) = \tau} \leq \frac 1{q-\Delta + 1}.
    $$
    The right hand side inequality trivially follows from the recursion in the proof of \Cref{lem:WSM_contraction}.
    Therefore, we have that
    \begin{align}\label{eq:base_case_WSM}
        \forall c\in [q], \abs{\Pr[T]{c(e) = c\mid c(\Lambda) = \tau} - \frac 1q} \leq \max\set{\frac{\Delta - 1}{\Delta}, \frac {e - 1}e}\frac 1q < \frac 1q
    \end{align}
    which serves as the base case of recursive contraction of the marginal probability.
    Let $\eps = \max\set{\frac{\Delta - 1}{\Delta}, \frac {e - 1}e}$.
    Therefore, Plugging \Cref{eq:base_case_WSM} and \Cref{lem:WSM_contraction}, we get that for $d\geq 2$,
    \begin{align}\label{eq:WSM_decay}
        \forall c\in [q], \abs{\Pr[T]{c(e) = c\mid c(\Lambda) = \tau} - \frac 1q} \leq \frac{\eps}{q} \tp{1 - \frac{1 - \eps}{2\Delta - (1+\eps)}}^{d-2}.
    \end{align}

    For general edge $e = \{u,v\}$, we split $e$ into two pendant edges $e_1 = \{u,w\}$ and $e_2 = \{w,v\}$ by adding a new vertex $w$ to $V$ and $\+L(e_1) = \+L(e_2) = \+L(e)$.
    Let $T'$ denote the new graph after splitting $e$. 
    Then we have that for any $c\in [q]$,
    \begin{align*}
        \Pr[T]{c(e) = c\mid c(\Lambda) = \tau} &= \Pr[T']{c(e_1) = c\mid c(e_1) = c(e_2), c(\Lambda) = \tau}
        \\&=\frac{\Pr[T']{c(e_1) = c(e_2) = c\mid c(\Lambda) = \tau}}{\Pr[T']{c(e_1) = c(e_2)\mid c(\Lambda) = \tau}}
        \\&= \frac{\Pr[T']{c(e_1) = c\mid c(\Lambda) = \tau}\Pr[T']{c(e_2) = c\mid c(\Lambda) = \tau}}{\sum_{c'\in[q]}\Pr[T']{c(e_1) = c'\mid c(\Lambda) = \tau}\Pr[T']{c(e_2) = c'\mid c(\Lambda) = \tau}}.
    \end{align*}
    The last equation follows from the disconnection between $e_1$ and $e_2$.
    For $d\geq 2$, let $\eta = \frac{\eps}{q} \tp{1 - \frac{1 - \eps}{2\Delta - (1+\eps)}}^{d-2}$ for simplicity.
    By \Cref{eq:WSM_decay}, we have that
    \begin{align*}
        \frac{(\frac 1q - \eta)^2}{q(\frac 1q + \eta)^2} \leq \Pr[T]{c(e) = c\mid c(\Lambda) = \tau} \leq \frac{(\frac 1q + \eta)^2}{q(\frac 1q - \eta)^2}.
    \end{align*}
    Therefore, for $d\geq 2$, plugging $\eta \leq \frac{\eps}{q}$ and the above inequality implies that 
    $$
        |\Pr[T]{c(e) = c\mid c(\Lambda) = \tau} - \frac 1q| \leq \frac{4\eta}{(1 - q\eta)^2}\leq \frac{4\eps}{q(1-\eps)^2} \tp{1 - \frac{1 - \eps}{2\Delta - (1+\eps)}}^{d-2}.
    $$
    We pick $C = \max\set{\frac{8\eps}{(1-\eps)^2}(1 - \frac{1 - \eps}{2\Delta - (1+\eps)})^{-2}, \tp{1 - \frac{1 - \eps}{2\Delta - (1+\eps)}}^{-1} }$ to finish the first part of the proof.

    For the second part, it is trivial after applying \Cref{eg:hardness_WSM}.
\end{proof}

%% file: mtd.tex
\section{Proofs for matrix trickle-down process}
\label{appendix-mtd}
In this section, we use the language of simplicial complexes as in \Cref{sec:covariance}.
That is, for a weighted edge coloring instance on a tree $T$ with $\beta$-extra
color losts $\+L$ and distribution $\mu$, we consider the colorings on a broom $K$ 
as a weighted simplicial complex $(\+C, \pi_{|K|})$ such that $\pi_{|K|} = \mu_K$.

Throughout this section, we consider $K$ as an edge set.
For a pinning $\tau\in \+C_{i}$ with $0\le i\le|K|-2$, we define
$K_\tau = \set{e\in K\mid e\notin \tau}$ and $\!{col}(\tau)=\set{c \mid \exists e\in K, \tau(e)=c}$. Let $\!{Id}_\tau \in \bb{R}^{\+C_1\times \+C_1}$ be the identity matrix restricted on $\+C_{\tau,1}$. Define 
$\!{Adj}_\tau \in \bb{R}^{\+C_1\times \+C_1}$ such that $\!{Adj}_\tau(ec,fc)=1$ if $ec,fc\in \+C_{\tau,1}$ and all other entries are $0$.
\subsection{Matrix trickle-down theorem}
The following proposition is the main tool we use in this section.
\begin{proposition}[Theorem 1.3 in \cite{ALOG21}]
	\label{thm:mtd-inductive}
	Given a pure $d$-dimensional weighted simplicial complex $(\+C,\pi_d)$, if there exists a family of matrices $\set{M_\tau\in\bb R^{\+C_1\times \+C_1}}$ satisfying
	\begin{itemize}
		\item For every $\tau\in\+C_{d-2}$, 
		\[
			\Pi_\tau P_\tau -2 \pi_\tau \pi_\tau^\top \mle M_\tau \mle \frac{1}{5}\Pi_\tau;
		\]
		\item For every face $\tau\in \+C_{d-k}$ with $k\geq 3$, one of the following two conditions hold:
		\begin{enumerate}
			\item 
			\[
				M_\tau \mle \frac{k-1}{3k-1}\Pi_\tau\quad\mbox{ and }\quad\E[x\sim\pi_\tau]{M_{\tau\cup\set{x}}}\mle M_\tau -\frac{k-1}{k-2}M_\tau\Pi^{-1}_\tau M_\tau .
			\]
			\item $(\+C_\tau,\pi_{\tau,k})$ is the product of $M$ pure weighted simplicial complexes $(\+C^{(1)},\pi^{(1)}), \dots (\+C^{(M)},\pi^{(M)})$ of dimension $n_1,\dots,n_M$ respectively and
			\[
				M_\tau = \sum_{i\in [M]\colon n_i\ge 2} \frac{n_i(n_i-1)}{k(k-1)}\cdot M_{\tau\cup \eta_{-i}}
			\]
			where $\eta_{-i} = \eta\setminus \+C^{(i)}_1$ for an arbitrary $\eta\in \+C_{\tau,k}$.
			\end{enumerate}
	\end{itemize}
	Then for every face $\tau\in\+C_{d-k}$ with $k\geq 2$, it holds that
	\[
		\Pi_\tau P_\tau -\frac{k}{k-1}\pi_\tau \pi_\tau^\top \mle M_\tau \mle \frac{k-1}{3k-1}\Pi_\tau.
	\]
	In particular, $\lambda_2(P_\tau)\le \lambda_1(\Pi_\tau^{-1}M_\tau)$.
\end{proposition}

\subsection{Base case}\label{sec:base-case}
\newcommand{\revp}{\tilde{\*p}}
We do matrix trickle down on $\mu_K(\tau)$.
Consider the base case: Assume the free edges in $K$ are $e,f$ and other edges are pinned with assignment $\tau$. 
Let $\*p_{ef}=(p_{e,c},p_{f,c})_{c\in \+L^\tau(e)\cup \+L^\tau(f)}$ and $\*p_{fe}=(p_{f,c},p_{e,c})_{c\in \+L^\tau(e)\cup \+L^\tau(f)}$. We omit the superscript $\tau$ in $q_e^\tau$ and $q_{e,f}^\tau$ in base case part since it is clear in the context.
Moreover, let $\revp_e=(0, p_{e, c})_{c\in \+L^\tau(e)\cup \+L^\tau(f)}, \revp_f = (p_{f, c}, 0)_{c\in \+L^\tau(e)\cup \+L^\tau(f)}$,
and define $\bb 1_e, \bb 1_f\in \bb R^{\+C_\tau, 1}$ such that $\bb 1_e(i, c) = \1{i=e},\bb 1_f(i, c) = \1{i=f}$.
Finally, we define $\*1_{ef}\in \bb R^{\+C_{\tau, 1}\times \+C_{\tau, 1}}$ such that $\*1_{ef}((i, c_1), (j, c_2) = \1{c_1=c_2\wedge i\neq j}$.
\begin{align*}
    &\phantom{{}={}}\Pi_\tau P_\tau- 2\pi_\tau\pi_\tau^\top 
    \\
    &= \frac{1}{2(q_eq_f-q_{e,f})^2}\-{diag}(\*p_{ef})\tp{(q_eq_f-q_{e,f})(\bb{1}_e\bb{1}_f^\top + \bb{1}_f\bb{1}_e^\top-\*1_{ef})-\frac{(q_e\bb{1}_f+q_f\bb{1}_e)(q_e\bb{1}_f+q_f\bb{1}_e)^\top}{2}+\*p_{fe}\*p_{fe}^\top} \-{diag}(\*p_{ef})\\
    &\mle \frac{1}{2(q_eq_f-q_{e,f})^2}\-{diag}(\*p_{ef})\tp{\*p_{fe}\*p_{fe}^\top-(q_eq_f-q_{e,f})\*1_{ef})} \-{diag}(\*p_{ef})\\
    &\mle \frac{1}{2(q_eq_f-q_{e,f})^2}\-{diag}(\*p_{ef})\tp{2(\revp_{e}\revp_{e}^\top+\revp_{f}\revp_{f}^\top)-(q_eq_f-q_{e,f})\*1_{ef})} \-{diag}(\*p_{ef}) .
\end{align*}
We do a row summation to bound the first term. Firstly,
\begin{align*}
    \revp_{e}\revp_{e}^\top+\revp_{f}\revp_{f}^\top \mle q_e\-{diag}(\revp_e) + q_f\-{diag}(\revp_f).
\end{align*}
Then,
\begin{align*}
    \frac{1}{(q_eq_f-q_{e,f})^2}\-{diag}(\*p_{ef})(\revp_{e}\revp_{e}^\top+\revp_{f}\revp_{f}^\top)\-{diag}(\*p_{ef})
    \mle \frac{1}{(q_eq_f-q_{e,f})^2}\-{diag}(\*p_{ef})(q_e\-{diag}(\revp_e) + q_f\-{diag}(\revp_f))\-{diag}(\*p_{ef}),
\end{align*}
which is a diagonal matrix.
On the entry $(ec, ec)$, it equals
\begin{align*}
    \frac{q_fp_{ec}}{(q_eq_f-q_{e,f})^2}p_{fc}p_{ec}
   &=       \frac{2q_fq_e}{q_fq_e}
      \cdot \frac{p_{ec}}{q_e}
      \cdot \frac{q_fq_e}{q_fq_e-q_{fe}}
      \cdot \frac{(q_f-p_{fc})p_{ec}}{2(q_fq_e-q_{fe})}
      \cdot \frac{p_{fc}}{q_f-p_{fc}}
    \\&\le
            2
      \cdot \frac1\beta
      \cdot \frac{\beta}{\beta-1}
      \cdot \pi_\tau(ec)
      \cdot \frac{1}{\beta-1}
      & \text{\big(\Cref{lem:marginal-bound-weighted-1,lem:marginal-bound-weighted-2}\big)}
    \\&=\frac{2}{(\beta-1)^2}\pi_\tau(ec).
\end{align*}
Applying the same argument to $(fc, fc)$, we have
\begin{align*}
    \Pi_\tau P_\tau- 2\pi_\tau\pi_\tau^\top 
    \mle
    -\frac{1}{2(q_eq_f-q_{e,f})^2}\-{diag}(\*p_{ef})\tp{(q_eq_f-q_{e,f})\*1_{ef})} \-{diag}(\*p_{ef})
    +
    \frac{2}{(\beta-1)^2}\Pi_\tau.
\end{align*}
Let $M_\tau$ be a block diagonal matrix with blocks $M_\tau^c$:
\begin{equation}\label{eqn:base-case-def-ol}
    M_\tau^c=\frac{p_{e,c}p_{f,c}}{2(q_e q_f-q_{e,f})}\tp{\begin{matrix}
    0 & -1\\
    -1 & 0
\end{matrix}}
+
\frac{2}{(\beta-1)^2}\Pi^c_\tau.
\end{equation}
Then we have the base case inequality
\begin{align*}
    \Pi_\tau P_\tau- 2\pi_\tau\pi_\tau^\top \mle \-{diag}(M_\tau^c).
\end{align*}

\subsection{Induction}
The induction step in \Cref{thm:mtd-inductive} is to show that for every $\tau$ with $\!{codim}(\tau) = k>2$ and connected $G_\tau$,
\begin{equation}\label{eqn:induction-main}
	\E[x\sim\pi_\tau]{M_{\tau\cup \set{x}}}\mle M_\tau - \frac{k-1}{k-2}M_\tau \Pi_\tau^{-1} M_\tau.
\end{equation}
For every $\tau$ and $c\in [q]$, we will define a matrix $M_\tau^c \in \bb R^{K^c\times K^c}$ and let $M_\tau$ be the block diagonal matrix with block $M_\tau^c$ for every $c\in [q]$. 
It is not hard to see that we only require 
\begin{equation}\label{eqn:induction-main-c}
	\E[x\sim\pi_\tau]{M_{\tau\cup \set{x}}^c}\mle M_\tau^c - \frac{k-1}{k-2}M_\tau^c (\Pi_\tau^c)^{-1} M_\tau^c
\end{equation}
to hold for every $c$ and $\tau$, where $(\Pi_\tau^c)^{-1}$ is $\Pi_\tau^{-1}$ restricted on $K^c\times K^c$.
We now describe our construction of $M^c_\tau$ for a fixed color $c$.
We write $M_\tau^c$ into the sum of a diagonal matrix and an off-diagonal matrix, i.e.,
\begin{equation}\label{eqn:N-decompose}
	M_\tau^c = \frac{1}{k-1}(A_\tau^c+\Pi_\tau^c B_\tau^c),
\end{equation}
where $A_\tau^c$ is an off-diagonal matrix and $B_\tau^c$ is a diagonal matrix.

From now on, when $c$ is clear from the context, we will omit the superscript $c$ for matrices.
For example, we will write $M_\tau$, $\Pi_\tau$, $\!{Adj}_\tau$, $\!{Id}_\tau$, $A_\tau$, $B_\tau$, $\dots$ instead of
                           $M^c_\tau$, $\Pi^c_\tau$, $\!{Adj}_\tau^{c}$, $\!{Id}_\tau^{c}$, $A_\tau^c$, $B_\tau^c$, $\dots$ respectively.
Plugging the above construction of $M_\tau$ into \eqref{eqn:induction-main-c} and remembering that the superscript $c$ has been omitted, we obtain
 \[
 (k-1)\cdot\E[x\sim\pi_\tau]{A_{\tau\cup\set{x}}+\Pi_{\tau\cup\set{x}}B_{\tau\cup\set{x}}} \mle (k-2)\cdot\tp{A_\tau+\Pi_\tau B_\tau} -\tp{A_\tau+\Pi_\tau B_\tau}\Pi_\tau^{-1}\tp{A_\tau+\Pi_\tau B_\tau}.
 \]
Here we need the following inequality of matrices.
\begin{lemma}[Corollary 12 in \cite{WZZ24}]
	\label{lem:matrix-squared-sum}
Let $A_1,\dots,A_n$ be a collection of symmetric matrices and $\Pi\mge 0$. Then
\[
	\tp{\sum_{i=1}^n A_i}\Pi \tp{\sum_{i=1}^n A_i} \mle n\sum_{i=1}^n A_i\Pi A_i.
\]	
\end{lemma}
 It follows from \Cref{lem:matrix-squared-sum} that
 \[
 \tp{A_\tau+\Pi_\tau B_\tau}\Pi_\tau^{-1}\tp{A_\tau+\Pi_\tau B_\tau}\mle 2A_\tau\Pi_\tau^{-1}A_\tau + 2\Pi_\tau B_\tau^2.
 \]
 As a result, in order for \eqref{eqn:induction-main} to hold, we only need to design $A_\tau$ and $B_\tau$  satisfying
\begin{equation}\label{eqn:condition-main}
 	(k-1)\cdot \E[x\sim\pi_\tau]{A_{\tau\cup\set{x}}} - (k-2) \cdot A_\tau+2A_\tau \Pi_\tau^{-1} A_\tau  \mle (k-2)\Pi_\tau B_\tau -(k-1)\cdot \E[x\sim\pi_\tau]{\Pi_{\tau\cup\set{x}} B_{\tau\cup\set{x}}}-2\Pi_\tau\tp{B_\tau}^2.
\end{equation}

\subsection{Construction of $A_\tau^i$}\label{sss:Atau}
We define
\begin{equation}\label{eqn:A-def}
	A_\tau = a_{k}\cdot (k-1)\cdot \E[\sigma\sim \pi_{\tau, k-2}]{A_{\tau\cup\sigma}}.
\end{equation}
Then we can deduce the following relation between $A_\tau$'s whose co-dimensions differ by one.
\begin{lemma}
	\[
		\E[x\sim \pi_\tau]{A_{\tau\cup\set{x}}} = \frac{k-2}{k-1}\cdot \frac{a_{k-1}}{a_{k}} A_\tau(ec, fc).
	\]
	where $h = h_\tau \geq 1$.
\end{lemma}
\begin{proof}
For any $ec, fc \in K_\tau$,
	\begin{align*}
		\E[x\sim \pi_\tau]{A_{\tau\cup\set{x}}}(ec, fc)
        &= (k-2)\sum_{x\in \+C_{\tau,1}} \frac{1}{k} \mu_{\ol K_x}^\tau(x) a_{k-1} \sum_{\sigma\in \+C_{\tau\cup\set{x}, k-3}} \frac{2}{(k-1)(k-2)}\mu^{\tau\cup\set{x}}_{\ol K_\sigma}(\sigma) A_{\tau\cup\set{x}\cup\sigma}(ec, fc)\\
		&= \frac{2}{k(k-1)} \sum_{x\in \+C_{\tau,1}}\sum_{\sigma\in \+C_{\tau\cup\set{x}, k-3}}\mu^{\tau}_{\ol K_x}(x) \mu^{\tau\cup\set{x}}_{\ol K_\sigma}(\sigma)a_{k-1}A_{\tau\cup\set{x}\cup\sigma}(ec, fc)\\
		&= \frac{2}{k(k-1)} \sum_{\sigma'\in \+C_{\tau,k-2}} \mu^\tau_{\ol K_{\sigma'}}(\sigma') A_{\tau\cup\sigma'}(uc, vc) (k-2)a_{k-1}\\
        &= (k-2)\E[\sigma'\sim \pi_{\tau, k-2}]{A_{\tau\cup\sigma'}(ec, fc)}\\
		&= \frac{k-2}{k-1}\cdot \frac{a_{k-1}}{a_{k}} A_\tau(ec, fc).
	\end{align*}
\end{proof}

It follows from the definition that $A_\tau$ is proportional to the expectation of the base cases when the boundary is drawn from $\pi_{\tau,k-2}$. For some technical reasons, we would like to isolate those boundaries containing the color $c$. This leads us to the following lemma.

\begin{lemma}~\label{lem:A-tau-i}
    \newcommand{\taur}{\tau\cup\omega|_{K_\tau\setminus\set{e, f}}}
	\[
		A_\tau = \frac{2a_k}{k}\sum_{\substack{\omega\in \+C_{\tau,k}\\ c\notin\!{col}(\omega)}}\pi_{\tau,k}(\omega)A^{i,\omega}_\tau,
	\]
	where $A^{i,\omega}_\tau$ is the matrix supported on $K_\tau^c\times K_\tau^c$ such that
	\[
		A_\tau^{\omega}(uc,vc) = -\frac{p^{\taur}_{e, c}p^{\taur}_{f, c}}{\ol q^{\taur}_{e}\ol q^{\taur}_{f} - \ol q^{\taur}_{ef}}, 
	\]
	and $\ol q^{\cdot}_{e}\defeq q^\cdot_{e} - p^\cdot_{ec}$,
    $\ol q^{\cdot}_{f}\defeq q^\cdot_{f} - p^\cdot_{fc}$, and
    $\ol q^{\cdot}_{ef}\defeq q^\cdot_{ef} - p^\cdot_{ec}p^\cdot_{fc}$.
\end{lemma}

\begin{proof}
    \newcommand{\taur}{\tau\cup\omega|_{K_\tau\setminus\set{e, f}}}
    In the proof we use $c\notin \cdot$ as a shortcut for $c\notin \!{col}(\cdot)$.
    By the definition of $A_\tau$,
    \begin{align*}
        A_\tau(ec, fc)
        &=  (k-1)a_k\frac{2}{k(k-1)}
              \sum_{\substack{\sigma\in \+C_{\tau, k-2}\\ c\notin\sigma}}
              \mu^\tau_{K_\tau\setminus\set{e, f}}(\sigma)
              \frac{p^{\tau\cup\sigma}_{ec}p^{\tau\cup\sigma}_{fc}}
              {q^{\tau\cup\sigma}_{e}q^{\tau\cup\sigma}_f - q^{\tau\cup\sigma}_{ef}}
      \\&=  \frac{2a_k}{k}
              \sum_{\substack{\sigma\in \+C_{\tau, k-2}\\ c\notin\sigma}}
              \Big(
              \sum_{\xi\in \+C_{\tau\cup\sigma, 2}}
              \mu^{\tau\cup\sigma}_{\set{e, f}}(\xi)
              \Big)
              \mu^\tau_{K_\tau\setminus\set{e, f}}(\sigma)
              \frac{p^{\tau\cup\sigma}_{ec}p^{\tau\cup\sigma}_{fc}}
              {q^{\tau\cup\sigma}_{e}q^{\tau\cup\sigma}_f - q^{\tau\cup\sigma}_{ef}}.
    \end{align*}
    The equality is becase if $c\in \sigma$, the $p^{\tau\cup\sigma}_{\cdot, c}$ terms vanish.
    Notice that for any $\sigma\in \+C_{\tau, k-2}$ and $\xi\in\+C_{\tau\cup\sigma, 2}$,
    \begin{align*}
        1 &= 
        \sum_{\xi\in \+C_{\tau\cup\sigma, 2}}
        \mu^{\tau\cup\sigma}_{\set{e, f}}(\xi)
        \\&=
        \frac{
        \sum_{\xi\in \+C_{\tau\cup\sigma, 2}}
        \mu^{\tau\cup\sigma}_{\set{e, f}}(\xi)
        }{
        \sum_{\substack{\xi\in \+C_{\tau\cup\sigma, 2}\\ c\notin \xi}}
        \mu^{\tau\cup\sigma}_{\set{e, f}}(\xi)
        }
        \sum_{\substack{\xi\in \+C_{\tau\cup\sigma, 2} \\ c\notin \xi}}
        \mu^{\tau\cup\sigma}_{\set{e, f}}(\xi)
        \\&=
        \frac{
        {q^{\tau\cup\sigma}_{e}q^{\tau\cup\sigma}_f - q^{\tau\cup\sigma}_{ef}}
        }{
        {\ol q^{\tau\cup\sigma}_{e}\ol q^{\tau\cup\sigma}_f - \ol q^{\tau\cup\sigma}_{ef}}
        }
        \sum_{\substack{\xi\in \+C_{\tau\cup\sigma, 2}\\ c\notin \xi}}
        \mu^{\tau\cup\sigma}_{\set{e, f}}(\xi).
    \end{align*}
    Multiply this expression in the former equation, we have
    \begin{align*}
        A_\tau(ec, fc)
        &=  \frac{2a_k}{k}
              \sum_{\substack{\sigma\in \+C_{\tau, k-2}\\ c\notin\sigma}}
              \mu^\tau_{K_\tau\setminus\set{e, f}}(\sigma)
              \frac{
              {q^{\tau\cup\sigma}_{e}q^{\tau\cup\sigma}_f - q^{\tau\cup\sigma}_{ef}}
              }{
              {\ol q^{\tau\cup\sigma}_{e}\ol q^{\tau\cup\sigma}_f - \ol q^{\tau\cup\sigma}_{ef}}
              }
              \sum_{\substack{\xi\in \+C_{\tau\cup\sigma, 2}\\ c\notin \xi}}
              \mu^{\tau\cup\sigma}_{\set{e, f}}(\xi)
                    \frac{p^{\tau\cup\sigma}_{ec}p^{\tau\cup\sigma}_{fc}}
                    {q^{\tau\cup\sigma}_{e}q^{\tau\cup\sigma}_f - q^{\tau\cup\sigma}_{ef}}
      \\&=  \frac{2a_k}{k}
              \sum_{\substack{\sigma\in \+C_{\tau, k-2}\\ c\notin\sigma}}
              \mu^\tau_{K_\tau\setminus\set{e, f}}(\sigma)
              \sum_{\substack{\xi\in \+C_{\tau\cup\sigma, 2}\\ c\notin \xi}}
              \mu^{\tau\cup\sigma}_{\set{e, f}}(\xi)
                    \frac{p^{\tau\cup\sigma}_{ec}p^{\tau\cup\sigma}_{fc}}
                         {\ol q^{\tau\cup\sigma}_{e}\ol q^{\tau\cup\sigma}_f - \ol q^{\tau\cup\sigma}_{ef}}
      \\&=  \frac{2a_k}{k}
              \sum_{\substack{\omega\in \+C_{\tau, k}\\ c\notin\sigma}}
              \mu^\tau_{K_\tau}(\omega)
                    \frac{p^{\taur}_{ec}p^{\taur}_{fc}}
                         {\ol q^{\taur}_{e}\ol q^{\taur}_f - \ol q^{\taur}_{ef}},
    \end{align*}
    and the lemma follows.
\end{proof}

\subsection{Spectral analysis of $A_\tau$}


For every $\omega\in\+C_{\tau,k}$ such that $c\not\in \!{col}((\omega\cup\tau)|_{K_\tau})$,
define $\Xi_\tau^{\omega} \in \bb R^{K^c\times K^c}$ as the diagonal matrix such that for every $ec \in K_\tau^{c}$:
\[
	\Xi_\tau^{\omega}(ec,ec) =
    \frac{p^{\tau\cup\omega|_{K\setminus \set e}}_{ec}}
         {\ol q^{\tau\cup\omega|_{K\setminus \set e}}_e}.
\]
\begin{lemma}\label{lem:A-tau-i-omega}
	\[
	A_\tau^{\omega} = \Xi_\tau^{\omega} (-\!{Adj}_\tau + \+R_\tau^{\omega}) \Xi_{\tau}^{\omega} ,
	\]
	where $\rho(\+R_\tau^{i,\omega}) \le  \frac{2(k-1)}{\beta - 1}$.
\end{lemma}
\begin{proof}
	Let $ec,fc\in K_\tau^{c}$.
    To ease the notation, when $\tau$ and $\omega$ are clear from the context,
    we use $\Gamma(e)$ and $\Gamma(e,f)$ to denote the partial coloring $(\tau\cup\omega)|_{V\setminus\set{e}}$ and $(\tau\cup\omega)|_{V\setminus\set{e,f}}$ respectively.
	Using our new notations, we have
	\[
	\Xi_\tau(ec,ec) =
    \frac{p^{\Gamma(e)}_{ec}}
         {\ol q^{\Gamma(e, f)}_e}.
	\]
	Observing that since $c\not\in\!{col}\tp{(\tau\cup \omega)|_{K_\tau}}$, we have
	\begin{align*}
        p^{\Gamma(e)}_{e, c} = p^{\Gamma(e, f)}_{e, c},
      \\p^{\Gamma(f)}_{f, c} = p^{\Gamma(e, f)}_{e, c}.
	\end{align*}
	So we can write $A_\tau^{i,\omega}$ as
	\begin{align*}
		A_\tau^{i,\omega}(uc,vc)
		&=-\frac{p^{\Gamma(e)}_{e, c}p^{\Gamma(f)}_{f, c}}
        {\ol q^{\Gamma(e,f)}_e \ol q^{\Gamma(e,f)}_f - \ol q^{\Gamma(e,f)}_{ef}}
	  \\&=\frac{p^{\Gamma(e)}_{e, c}p^{\Gamma(f)}_{f, c}}
        {\ol q^{\Gamma(e)}_e \ol q^{\Gamma(f)}_f}
        \Bigg(
        -1 + \+R_\tau^{\omega}(ec,fc) 
        \Bigg)
	\end{align*}
	where 
	\begin{equation}\label{eqn:eta-remain}
		\+R_\tau^{i,\omega}(ec,fc) =
        \frac{
        -{\ol q^{\Gamma(e)}_e \ol q^{\Gamma(f)}_f}
        +
        {\ol q^{\Gamma(e,f)}_e \ol q^{\Gamma(e,f)}_f - \ol q^{\Gamma(e,f)}_{ef}}
        }{
        {\ol q^{\Gamma(e,f)}_e \ol q^{\Gamma(e,f)}_f - \ol q^{\Gamma(e,f)}_{ef}}
        }.
	\end{equation}
    Notice that
    \begin{align*}
    \abs{\ol q^{\Gamma(e,f)}_e \ol q^{\Gamma(e,f)}_f - \ol q^{\Gamma(e)}_e \ol q^{\Gamma(f)}_f}
    &=
    \abs{\ol q^{\Gamma(e,f)}_e p^{\Gamma(e, f)}_{f, \omega(e)} +
         \ol q^{\Gamma(e,f)}_f p^{\Gamma(e, f)}_{e, \omega(f)} -
         p^{\Gamma(e, f)}_{f, \omega(e)} p^{\Gamma(e, f)}_{e, \omega(f)} }
    \\&\le 
         \ol q^{\Gamma(e,f)}_e p^{\Gamma(e, f)}_{f, \omega(e)} +
         \ol q^{\Gamma(e,f)}_f p^{\Gamma(e, f)}_{e, \omega(f)}.
    \end{align*}
    Therefore,
	\begin{align*}
		\abs{\+R_\tau^{i,\omega}(ec,fc)} \le
        \frac{
         \ol q^{\Gamma(e,f)}_e p^{\Gamma(e, f)}_{f, \omega(e)} +
         \ol q^{\Gamma(e,f)}_f p^{\Gamma(e, f)}_{e, \omega(f)}
        }{
        \ol q^{\Gamma(e,f)}_e \ol q^{\Gamma(e,f)}_f - \ol q^{\Gamma(e,f)}_{ef}
        },
	\end{align*}
    which equals 
    $\mu^{\Gamma(e, f)}_{\set e;\+L'}(\omega(e)) + \mu^{\Gamma(e, f)}_{\set f;\+L'}(\omega(f))$,
    defining $\+L'$ be the color lists obtained by removing $c$ from the color lists of $e, f$.
    By \Cref{lem:marginal-bound-weighted-1}, this is bounded by $\frac{2}{\beta-1}$
    since the modified coloring istance is $(\beta-1)$-extra.
    The lemma then follows by doing a row summation to $\+R_\tau^\omega$.
\end{proof}

\begin{proposition}\label{prop:matrix-sq-coeff}
    Consider non-zero coefficients $\lambda_i$, $i\in [N]$, matrix $A = \sum_{i\in [N]}\lambda_i A_i$ and 
    $\Sigma\mle 0$, and all matrices are square, of the same size and symemtric. Then
    \[
    A\Sigma A \mle
    \big(\sum_{i\in [N]} \lambda_i\big)
    \bigg( \sum_{i\in [N]}\lambda_i A_i\Sigma A_i \bigg).
    \]
\end{proposition}
\begin{proof}
    \begin{align*}
    A\Sigma A &= \sum_{i, j\in [N], i < j} \lambda_i\lambda_j \big( A_i\Sigma A_j + A_i\Sigma A_j \big)
               + \sum_{i\in [N]}\lambda_i^2 A_i\Sigma A_i
    \\& \mle \sum_{i, j\in [N]} \lambda_i\lambda_j\big( A_i\Sigma A_i + A_j\Sigma A_j \big)
               + \sum_{i\in [N]}\lambda_i^2 A_i\Sigma A_i
    \\& = 
    \big(\sum_{i\in [N]} \lambda_i\big)
    \bigg( \sum_{i\in [N]}\lambda_i A_i\Sigma A_i \bigg).
    \end{align*}
\end{proof}

Define $C_\tau\defeq \sum_{\substack{\omega\in \+C_\tau} \\ c\notin \omega}\pi_\tau(\omega)$,
and $\+C'_{\tau, k}\defeq \set{\omega\in \+C_{\tau, k} \mid c\notin\omega}$,
then the following lemma holds.
\begin{lemma}\label{lem:bound-A-tau-i}
 $A_\tau\Pi_\tau^{-1}A_\tau
 \mle
 \frac{4a_k^2C_\tau}{k^2} \sum_{\omega\in\+C'_{\tau, k}} \pi_\tau(\omega) A_\tau^{\omega}\Pi^{-1}_\tau A_\tau^{\omega}
 \mle
 \frac{4a_k^2}{k^2} \sum_{\omega\in\+C'_{\tau, k}} \pi_\tau(\omega) A_\tau^{\omega}\Pi^{-1}_\tau A_\tau^{\omega}$.
\end{lemma}
\begin{proof}
    By \Cref{lem:A-tau-i-omega} and \Cref{prop:matrix-sq-coeff}.
\end{proof}

In the following discussion,
let $\gamma=\bigg( 1+ \frac{\Delta-1}{\beta-1} \bigg)^3\frac{1}{\beta-1}$.
\begin{lemma}\label{lem:bound-A-tau-i-omega-square}
	$ A_\tau^{\omega} \Pi_\tau^{-1} A_\tau^{\omega}
    \mle \gamma k\cdot \Xi_\tau^{\omega}\tp{\tp{\!{Adj}_\tau}^2
    +
    \frac{4(k-1)^2}{(\beta-1)^2}\!{Id}_\tau}\Xi_\tau^{i,\omega}$.
\end{lemma}
\begin{proof}
	By \Cref{lem:marginal-lower-weighted,lem:marginal-upper-weighted,lem:marginal-bound-weighted-1} for any $ec \in K_\tau^{c}$,
	\begin{align*}
			\Xi_\tau^{\omega}\Pi_\tau^{-1}(uc, uc)\leq
            k \frac{(\beta+k-2)(\beta+\Delta-2)}{(\beta-1)^2}\frac{\ell_u - k + 1}{\ell_u-k-\Delta + 1}
            &\le
            k \frac{(\beta+\Delta-2)^2}{(\beta-1)^2}\frac{\beta + \Delta-2}{\beta-1}
            \\&\le k\bigg( 1+ \frac{\Delta-1}{\beta-1} \bigg)^3.
	\end{align*}
	Then it follows from \Cref{lem:A-tau-i-omega} and $\Xi_\tau^{i,\omega}\mle \frac{1}{\beta-1}\cdot \!{Id}_{\tau}$ that
    \[
	 A_\tau^{\omega} \Pi_\tau^{-1} A_\tau^{\omega}
     \mle
     2 k\gamma \Xi_\tau^{i,\omega}\tp{\tp{\!{Adj}_\tau}^2  + \frac{4(k-1)^2}{(\beta-1)^2}\!{Id}_\tau}\Xi_\tau^{\omega}.
     \qedhere
    \]

\end{proof}

\begin{lemma}\label{lem:bound-Xi-Pi}
$ \frac{1}{k}\sum_{\omega\in \+C_{\tau, k}}\pi_{\tau, k}(\omega)\Xi^\omega_\tau\Pi^{-1}_\tau = \!{Id}_\tau$.
\end{lemma}
\begin{proof}
    At entry $(ec, ec)$ such that $\pi_\tau(ec)\neq 0$, the LHS is
	\begin{align*}
        \frac{1}{k} \sum_{\omega\in \+C_{\tau, k}}\mu^\tau_{K_\tau}(\omega) \pi^{-1}_\tau(ec)
                      \frac{p^{\tau\cup\omega|_{K\setminus \set e}}_{e, c}}
                       {\ol q^{\tau\cup\omega|_{K\setminus \set e}}_{e   }}
        &= 
        \frac{1}{k} \sum_{\omega'\in \+C_{\tau, K_\tau}}\mu^\tau_{K_\tau}(\omega') \pi^{-1}_\tau(ec)
                      \frac{p^{\tau\cup\omega'}_{e, c}}
                       {\ol q^{\tau\cup\omega'}_{e   }}
      \\&= 
        \frac{1}{k} \sum_{\omega'\in \+C_{\tau, K_\tau}}\mu^\tau_{K_\tau\setminus\set e}(\omega')
                      \mu^{\tau\cup\omega'}_{\set e}(c)
                    \pi^{-1}_\tau(ec)
      \\&= 
        \frac{1}{k} \sum_{\omega'\in \+C_{\tau, K_\tau}}\mu^\tau_{K_\tau}(\omega\cup\set{ec})
                    \pi^{-1}_\tau(ec)
         = 
        \pi^{-1}_\tau(ec) \pi_\tau(ec)
        =1.
	\end{align*}
\end{proof}

We are now ready to bound $\tp{(k-1)\cdot \E[x\sim\pi_\tau]{A_{\tau\cup\set{x}}} - (k-2) \cdot A_\tau+4A_\tau \Pi_\tau^{-1} A_\tau}$ in the LHS of \Cref{eqn:condition-main}.
\begin{lemma}
	\label{lem:bound-A-tau}
	There exists a sequence of non-negative numbers $\set{a_h}_{0\le h\le \Delta}$ such that
	\[
		\Pi_\tau^{-\frac{1}{2}} \tp{(k-1)\cdot \E[x\sim\pi_\tau]{A_{\tau\cup\set{x}}}
        - (k-2) \cdot A_\tau+2A_\tau \Pi_\tau^{-1} A_\tau}
        \Pi_\tau^{-\frac{1}{2}} \mle
        \frac{8\gamma(k-1)}{\beta-1}\Big(1 + \frac{\Delta}{\beta-1}\big(1+\frac{2}{\beta-1}\big)\Big).
	\]
\end{lemma}
\begin{proof}
	Applying \Cref{lem:A-tau-i} and \Cref{lem:bound-A-tau-i}, we obtain
\[
	\!{LHS} \mle 
    \Pi_\tau^{-\frac{1}{2}}
    \tp{
    \frac{2(k-2)}{k}
    (a_{k-1}-a_k)
    \sum_{\omega\in\+C'_{\tau,k}}
      A_\tau^{\omega}
    +
    \frac{8 a_k^2}{k^2}\cdot 
    \sum_{\omega\in \+C'_{\tau,k}}
      A_\tau^{\omega}\Pi^{-1}_\tau A_\tau^{\omega}
    }\Pi_\tau^{-\frac{1}{2}}.
\]
Then by \Cref{lem:A-tau-i-omega} and \Cref{lem:bound-A-tau-i-omega-square}, we can bound above by
\begin{align}
	\label{eq:lem-bound-A-tau-1}
	\!{LHS} \mle\frac{2}{k}
    \Pi_\tau^{-\frac{1}{2}}
    \Xi_\tau^{\omega}
    \Bigg(
    &(k-2)
    (a_{k-1}-a_k)
    \sum_{\omega\in\+C'_{\tau,k}}
      \pi_\tau(\omega)
      (-\!{Adj}_\tau + \frac{2(k-1)}{\beta-1}\!{Id}_\tau)
    \\+&
    4\gamma a_k^2 \cdot
    \sum_{\omega\in \+C'_{\tau,k}}
      \pi_\tau(\omega)
      \Big(
      \!{Adj}^2_\tau + \frac{4(k-1)^2}{(\beta-1)^2}\!{Id}_\tau
      \Big)
    \Bigg)
    \Xi_{\tau}^{\omega}
    \Pi_\tau^{-\frac{1}{2}},
\end{align}
where $\gamma = \bigg( 1+ \frac{\Delta-1}{\beta-1} \bigg)^3\frac{1}{\beta-1}$.

We want to find a sequence of $\set{a_k}$
so that the spectral radius of the following matrices
$\tilde A_k$ appearing in the non-remainder terms in \Cref{eq:lem-bound-A-tau-1} is small:
\[
	\tilde A_h\defeq -(k-2)(a_{k}-a_{k-1})\!{Adj}_\tau + 4a_k^2\gamma\tp{\!{Adj}_\tau}^2.
\]
Since the spectrum of $\!{Adj}_\tau$ is $\set{-1, (k-1)}$, the spectrum of $\tilde A_k$ is
\[
	\set{
		(k-2)(a_{k-1}-a_k) + 4\gamma a_k^2, \;
		-(k-1)(k-2)(a_{k-1}-a_k) + 4\gamma (k-1)^2 a_k^2
	}.
\]
Define
\[
	a_k = \frac{1}{1+4\gamma(k-2)} (2\leq k \leq \Delta).
\]
Then we have
\begin{equation}\label{eqn:spectral-radius-of-A}
	\rho(\tilde A_k)\le \frac{4 \gamma  (1+4\gamma(k-2)) h}{(4 \gamma  (k-3)+1) (4 \gamma  (k-2)+1)^2}\le 4\gamma (k-1)
\end{equation}
for $k \geq 3$. In particular, when $k=2$, $\rho(\tilde A_h)\leq 4a_k^2 \gamma (k-1)^2 = 4\gamma$,
which is consistent with the above bound.
So we have $\rho(\tilde A_k)\leq 4\gamma (k-1)$ for $h\geq 1$.
Note that $\Xi_\tau^{\omega}\mle \frac{1}{\beta-1}\cdot \!{Id}_{\tau}$,
it then follows from \Cref{eqn:spectral-radius-of-A} and \Cref{lem:bound-Xi-Pi} that 
\begin{align}\label{eqn:bound-adj}
	&\phantom{{}={}}
	\frac{2}{k}
    \Pi_\tau^{-\frac{1}{2}}
    \Xi_\tau^{\omega}
    \Bigg(
    -(k-2)
    (a_{k-1}-a_k)
    \sum_{\omega\in\+C'_{\tau,k}}
      \pi_\tau(\omega)
      \!{Adj}_\tau
    +
    4\gamma a_k^2
    \sum_{\omega\in \+C'_{\tau,k}}
      \!{Adj}^2_\tau
    \Bigg)
    \Xi_{\tau}^{\omega}
    \Pi_\tau^{-\frac{1}{2}}
    \notag
    \\&\le \frac{8\gamma (k-1)}{\beta-1} \!{Id}_\tau.
\end{align}
A direct calculation yields that
\begin{align}\label{eqn:bound-remainder}
    &\phantom{{}={}}
	\frac{2}{k}
    \Pi_\tau^{-\frac{1}{2}}
    \Xi_\tau^{\omega}
    \Bigg(
    (k-2)
    (a_{k-1}-a_k)
    \sum_{\omega\in\+C'_{\tau,k}}
      \pi_\tau(\omega)
      \frac{2(k-1)}{\beta-1}\!{Id}_\tau
    +
    4\gamma a_k^2 \cdot
    \sum_{\omega\in \+C'_{\tau,k}}
      \pi_\tau(\omega)
      \frac{4(k-1)^2}{(\beta-1)^2}\!{Id}_\tau
    \Bigg)
    \Xi_{\tau}^{\omega}
    \Pi_\tau^{-\frac{1}{2}}
    \notag
    \\&\mle
    \frac{8\gamma(k-1)^2}{(\beta-1)^2}\Big(\frac{2}{\beta-1} + 1\Big)\!{Id}_\tau
    \mle
    \frac{8\gamma(k-1)\Delta}{(\beta-1)^2}\Big(\frac{2}{\beta-1} + 1\Big)\!{Id}_\tau.
\end{align}
Combining \Cref{eqn:bound-adj} and \Cref{eqn:bound-remainder} finishes the proof.
\end{proof}

\subsection{Construction of $B_\tau$}
\label{sss:Btau}
For $\tau$ of co-dimension $k > 2$,
we introduce coefficients $\set{b_k}_{3 \leq h \leq \Delta}$
whose values will be determined later, and define $B_\tau$ as follows:
\begin{equation} \label{eqn:B-def}
	B_\tau(ec,ec) = b_{k}
\end{equation}
for any $e \in K_\tau$ and all other entries are $0$
When $k=2$, this is exactly the base case considered in \Cref{sec:base-case}.
According to~\eqref{eqn:base-case-def-ol}, we have $b_1=\frac1{\tp{\beta-1}^2}$.

Notice that if $\!{codim}(\tau) = k\ge 3$, then
\begin{align}
\E[x\sim\pi_\tau]{\Pi_{\tau\cup\set{x}} b_k}
=
\pi_\tau(ec)^{-1}\sum_{x\in \+C_\tau}{\pi_\tau(x)\pi_{\tau\cup \set{x}}(ec)} b_{k-1}
=
\sum_{x\in \+C_\tau}{\pi_{\tau\cup \set{ec}}(x) b_{k-1}} = b_{k-1}.
\label{eqn:contraint-B-entry}
\end{align}
where the second equality follows from the fact that $\pi_\tau(x)\pi_{\tau\cup \set{x}}(vc)=\pi_{\tau\cup\set{vc}}(x)\pi_\tau(vc)$. 

By the discussion in the last section and the above definition, now \cref{eqn:condition-main} 
becomes
\[
(k-2)b_k - (k-1)b_{k-1} - 2b_k^2 \ge 
    \frac{8\gamma(k-1)\Delta}{(\beta-1)^2}\Big(\frac{2}{\beta-1} + 1\Big)
\]
for all $3\le k \le \Delta$.

Assume $\beta \geq 11$. Since 
\begin{align}\label{eqn:A-tau-i-upper-bound}
\begin{split}
		\Pi_\tau^{-1/2}A_\tau\Pi_\tau^{-1/2}
		&=  \frac{2a_{k}}{k}\cdot \Pi_\tau^{-1/2} \sum_{\omega\in \+C'_{\tau,k}}\pi_\tau(\omega) A^{\omega}_\tau \Pi_\tau^{-1/2} \\
		&=  \frac{2a_{k}}{k}\cdot \Pi_\tau^{-1/2} \sum_{\omega\in \+C_{\tau,k}'}\pi_\tau(\omega) \Xi_\tau^{\omega}(-\!{Adj}_\tau +\+R_\tau^{\omega})\Xi_\tau^{\omega}\Pi_\tau^{-1/2}\\
		&\mle \frac{2a_{k}}{(\beta-1)}\tp{1+\frac{2(k-1)}{\beta-1}}\!{Id}_\tau\\
		&\mle \frac{1}{(\beta-1)}\frac{1+\frac{2(k-1)}{\beta-1}}{1+\frac{4(k-2) \tp{1+\frac{\Delta-1}{\beta-1}}^3} {\beta-1}}\!{Id}_\tau\\
        &\mle \frac{1}{\beta-1} \!{Id}_\tau\\
		&\mle \frac{1}{10}\!{Id}_\tau,
\end{split}
\end{align}
we have $M_\tau=\frac{\sum_i A_\tau + \Pi_\tau B_\tau}{k-1} \mle \frac{k-1}{3k-1}\Pi_\tau$ as long as $B_\tau \mle \tp{\frac{(k-1)^2}{3k-1}-\frac{1}{10}}\!{Id}_\tau$.
We strengthen this constraint to $B_\tau \mle \tp{\frac{1}{5}-\frac{1}{10}}\!{Id}_\tau=\frac{1}{10}\!{Id}_\tau$.

For brevity, we denote ${8\gamma\Delta}\Big(\frac{2}{\beta-1} + 1\Big)$ by $C(\Delta)$ in the following calculation. Therefore, our constraints for $\set{b_k}_{1\leq k\leq \Delta}$ are
\begin{equation}
	\label{eqn:constraint-B}\tag{$\blacktriangle$}
	\begin{cases}
		(k-2) b_k - (k-1)b_{k-2} \ge 2b_k^2+\frac{C(\Delta)(k-1)}{(\beta-1)^2},& 3\le k\le \Delta; \\
		b_k \le \frac1{10}, & 2\le k\le \Delta.
	\end{cases}
\end{equation}
It follows from Lemma 26 in~\cite{WZZ24} that there is a feasible solution of \cref{eqn:constraint-B}
as long as $\beta\ge c\sqrt\Delta\log^2\Delta + 2c$, where $c = \sqrt{20(1+2C(\Delta))}$. And the solution $b_k \leq \frac{1}{(\beta-1)^2}\tp{1+(6+16 C(\Delta))\Delta \log^2 \Delta}$.
Notice that if $\beta - 1\ge \max\set{\Delta, 10}$, then $C(\Delta)\le \frac{10}{\Delta-1}$,
and the solution exists if $\beta - 1\ge 20\log^2\Delta+2\sqrt{\frac{200}{\Delta-1}}$.

Putting all constraints to $\beta$ together, we have
\[\beta-1\ge \max\set{\Delta, 10, 20\log^2\Delta+2\sqrt{\frac{200}{\Delta-1}}},\]
which can be unified to a single bound that $\beta\ge \Delta + 50$.

\subsection{Proof of \Cref{lem:PiP-bound}}
\begin{proof}[Proof of \Cref{lem:PiP-bound}]
For any $\tau \in \+C_{d-k}$ with $k \geq 2$, we construct the matrices $A_\tau$ and $B_\tau$ as \Cref{sss:Atau} and \Cref{sss:Btau}. Then we have
    \[ \rho(\Pi_\tau^{-1/2}A_\tau\Pi_\tau^{-1/2})+ \rho(B_\tau) \le \frac{1}{\beta-1}+\frac{1}{(\beta-1)^2} + \frac{(6+16 C(\Delta))\Delta \log^2\Delta}{(\beta-1)^2}.
    \]
When $\beta \geq \Delta + 50$, the above term is upper bounded by $\eta_\Delta \defeq \frac{1+(6+\frac{160}{\Delta-1})\log^2 \Delta}{\Delta}+\frac{1}{\Delta^2}$.
Applying \Cref{thm:mtd-inductive}, we have
\[
\rho( P_\tau -\frac{k}{k-1}\*{1} \pi_\tau^\top) \le \rho(\Pi_\tau^{-1} M_\tau) \le \frac{\eta_{\Delta}}{k-1}.
\]
Taking $\tau = \emptyset$, we obtain that
\[
\Pi P - \frac{d}{d-1} \pi_\tau \pi_\tau^\top \mle \frac{\eta_{\Delta}}{d-1}\Pi.
\qedhere
\]
\end{proof}